\documentclass[11pt]{article}
\usepackage[hmargin=1in,vmargin=1in]{geometry}
\usepackage{hchang}

\usepackage{comment}
\usepackage{xspace}
\usepackage{multirow}
\usepackage[hyperref]{ntheorem}
\usepackage{cleveref}
\usepackage{thm-restate}

\nolinenumbers

\theoremseparator{.}
\theorembodyfont{\slshape}
\newtheorem{theorem}{Theorem}[section]
\newtheorem{lemma}[theorem]{Lemma}

\newtheorem{observation}[theorem]{Observation}
\newtheorem{problem}[theorem]{Problem}

\theorembodyfont{\upshape}
\newtheorem{remark}[theorem]{Remark}
\newtheorem{definition}[theorem]{Definition}

\def\eps{\varepsilon}

\def\bdry{\partial\!}
\newcommand{\bd}{\partial\!}

\def\reals{\mathbb{R}}

\newcommand{\dist}{\mathrm{dist}}

\newcommand{\Rep}{\mbox{\sf Rep}}

\newcommand{\vcdim}{\mbox{\sf VC-Dim}}

\newcommand{\OO}{\tilde{O}}

\newcommand{\cube}[1]{\llbracket #1\rrbracket}

\newcommand{\Diameter}{\textsc{Diameter}\xspace}

\newcolumntype{L}{>{$}l<{$}}

\renewcommand{\arraystretch}{1.3}

\newcounter{ctr}
\loop
 \stepcounter{ctr}
 \expandafter\edef\csname c\Alph{ctr}\endcsname{\noexpand\mathcal{\Alph{ctr}}}
\ifnum\thectr<26
\repeat

\renewcommand{\note}[1]{}

\title{Charting the Diameter Computation Landscape on Intersection Graphs in the Plane}
\date{}
\author{Timothy M. Chan%
    \thanks{Siebel School of Computing and Data Science,
    University of Illinois at Urbana-Champaign. Email: tmc@illinois.edu.  Supported by NSF grant CCF-2224271.}
    \and
    Hsien-Chih Chang%
    \thanks{Department of Computer Science, Dartmouth College. Email: hsien-chih.chang@dartmouth.edu.}
    \and
    Jie Gao%
    \thanks{Department of Computer Science, Rutgers University. Email: jg1555@rutgers.edu. Gao would like to acknowledge NSF support through CNS-2515159, IIS-2229876, DMS-2220271,  DMS-2311064, CCF-2208663, CCF-2118953.}
    \and
    S\'andor Kisfaludi-Bak%
    \thanks{Department of Computer Science, Aalto University, Finland. Email: 
    sandor.kisfaludi-bak@aalto.fi. Supported by the Research Council of
    Finland, Grant 363444.}
    \and
    Hung Le%
    \thanks{Manning CICS, UMass Amherst. Email: hungle@cs.umass.edu. Supported by NSF grants CCF-2517033 and CCF-2121952, NSF CAREER Award CCF-2237288, and a Google Faculty Research Award.}  
    \and 
    Da Wei Zheng%
    \thanks{Institute of Science and Technology Austria, Klosterneuburg, Austria. Work in this paper was started while at the Siebel School of Computing and Data Science, University of Illinois at Urbana-Champaign.
    This project has received funding from the Austrian Science Fund (FWF) grant  \href{https://www.doi.org/10.55776/I5982}{DOI 10.55776/I5982}. For open access purposes, the author has applied a CC BY public copyright license to any author-accepted manuscript version arising from this submission.}
}

\begin{document}
\maketitle

\begin{abstract}
Computing the diameter of the intersection graphs of objects is a basic problem in computational geometry. Previous works showed that the complexity of computing the diameter mainly depends on the \emph{object types}: for unit disks and squares in 2D, the problem is solvable in truly subquadratic time~\cite{ChanCGKLZ25}, while for other objects, including unit segments and equilateral triangles in 2D or unit balls and axis-parallel unit cubes in 3D, there is no truly subquadratic time algorithm under the Orthogonal Vector (OV) hypothesis~\cite{Bringmann2022-me}.   

We undertake a comprehensive study of computing the diameter of geometric intersection graphs for various types of objects. We discover many new irregularities, showing that the landscape is extremely nuanced: the source of hardness is a combination of \emph{the object type, the true diameter value, and how the objects intersect with each other}. 
Our highlighted results for the 2D case include:
\begin{enumerate}
    \item The diameter of non-degenerate, axis-aligned line segments can be computed in truly subquadratic time. Previous hardness result~\cite{Bringmann2022-me} for line segments applies only to degenerate instances. 
    On the other hand, for the degenerate case, we show that a truly subquadratic time algorithm exists when the true diameter is constant. 
    \item An almost-linear-time algorithm for unit-square graphs of \emph{constant diameter}. Previous algorithms~\cite{duraj2023better,ChanCGKLZ25} rely on succinct representation assuming bounded VC-dimension; for such a strategy $\Omega(n^{7/4})$ time is an inherent barrier. 
    \item An $\OO(n^{4/3})$-time algorithm to decide if the diameter of a unit-disk graph is at most 2. This improves upon the recent algorithm with running time  $\OO(n^{2-1/9})$~\cite{ChanCGKLZ25}.
    \item Deciding if the diameter of intersection graphs of fat triangles or line segments is at most 2 is truly subquadratic-hard under fine-grained complexity assumptions.  Previous lower bounds~\cite{Bringmann2022-me} only hold when deciding if diameter is at most 3.
\end{enumerate}
Our findings are presented in a pair of papers. 
This paper focuses solely on the 2D case, while the companion paper is devoted to higher-dimensional cases.  
\end{abstract}
\thispagestyle{empty}

\newpage
\tableofcontents
\thispagestyle{empty}

\newpage
\setcounter{page}{1}
\section{Introduction}

Computing the diameter of sparse graphs is \EMPH{truly subquadratic hard}: under the Strong Exponential Time Hypothesis (SETH), there is no $O(n^{2-\varepsilon})$ algorithm for 
distinguishing diameter 2 vs.\ 3 of graphs with $n$ vertices and $\OO(n)$\footnote{$\OO(\cdot)$ hides polylogarithmic  factors.} 
edges~\cite{Roditty2013-zr}. This conditional lower bound implies that the trivial algorithm for computing the diameter by doing BFS from every vertex is essentially optimal. Since then, the major research focus has been on substantially beating the BFS-based algorithm for structural families of graphs, especially planar/minor-free graphs and geometric intersection graphs.  For planar/minor-free graphs, truly subquadratic algorithms were known~\cite{Cabello2018-gz,Gawrychowski2018-zy, Ducoffe2019-tp,le2023vc,ChanCGKLZ25}.  

For geometric intersection graphs, the complexity of diameter computation remains poorly understood due to the sheer diversity of geometric objects underlying the intersection graphs, for example, line segments, strings, (unit) disks, triangles, rectangles, balls, hypercubes, (axis-aligned) boxes, to name just a few. Another important subtlety is that intersection graphs can have $\Omega(n^2)$ edges; therefore, a truly subquadratic time algorithm must use an implicit graph representation. (In a geometric intersection graph, each vertex corresponds to a geometric object, and there is an edge between every two intersecting objects.)  

The pioneering work by Bringmann \etal~\cite{Bringmann2022-me} studied computing the diameter of intersection graphs of several types of objects: axis-parallel unit segments, equilateral triangles, axis-parallel hypercubes, and balls. They showed that it is truly subquadratic hard to solve: (i)  \Diameter-$2$ for axis-parallel hypercubes in $\reals^{12}$ under the Hyperclique Hypothesis; (ii) \Diameter-$3$ for (not necessarily axis-aligned) unit segments and equilateral triangles in $\reals^2$ under Orthogonal Vector (OV) Hypothesis;  and (iii)  \Diameter-$\Omega(\log n)$ for unit balls and axis-parallel unit cubes in $\reals^3$, and axis-parallel line segments in $\reals^2$, all under OV Hypothesis. Here  \Diameter-$\Delta$ asks to decide if a given input graph has diameter at most $\Delta$  or at least $\Delta+1$. 
Notably, basic objects such as axis-aligned squares, disks, and their unit versions were absent from this list of lower bounds. However, the same paper showed that \Diameter-2 for axis-aligned unit squares can be solved in $O(n\log n)$ time, which hinted at the possibility of truly subquadratic time for other diameter values and objects such as axis-aligned squares and (unit) disks.  Later,  Duraj, Konieczny, and Pot{\k e}pa~\cite{duraj2023better} gave an algorithm with running time $\OO(\Delta n^{7/4})$ for  \Diameter-$\Delta$  of axis-aligned unit squares: it is truly subquadratic when the true diameter is $O(n^{1/4-\varepsilon})$. Chang, Gao, and Le~\cite{CGL24} (arXiv version) designed a truly subquadratic $+1$-approximation algorithm for the diameter of axis-aligned unit squares, unit disks, and more generally, similar-sized fat pseudo-disks.  Recently, Chan \etal~\cite{ChanCGKLZ25} gave algorithms for computing the exact diameter of unit disks and arbitrary axis-aligned squares; the running time is truly subquadratic regardless of the diameter value.  

Altogether, the previous work reveals highly contrasting pieces of the complexity landscape of geometric intersection graphs which depend on \EMPH{object types}: either computing a very small diameter (e.g., $\Delta=2$ or $3$) is truly subquadratic hard for some objects, or the diameter can be computed in truly subquadratic time regardless of the value of the diameter for others. 

We undertake a comprehensive study of computing the diameter of geometric intersection graphs for various types of objects. Our results are presented in two papers. The present paper 
focuses solely on the 2D case. 
All the previous truly subquadratic algorithms for diameter are for objects on the plane, and hence this basic case is of particular interest. We present numerous results, both positive and negative, on computing the diameter of the geometric intersection graphs of a wide variety of 2D objects. A high-level takeaway from our collection of results is that the complexity of computing the diameter depends not only on object type, but also on  \EMPH{diameter range and how the objects intersect}.

In the concurrent paper~\cite{Anon3D25}, we study objects in dimensions at least $3$ and present various upper and lower bounds.  Both papers together show that the landscape of computing the diameter on geometric intersection graphs is much more nuanced than previously known.

\begin{table}[!t]
\footnotesize\sffamily
\renewcommand{\arraystretch}{1.3}
\centering
\begin{NiceTabular}{cccc}
\multicolumn{2}{c}{\textbf{Graph class}}  & \textbf{Lower bound} & \textbf{Upper bound}  \\
\hline
\hline
\multirow{1}{*}{Unit disks}
  &  &  & \shortstack{$O^*(n^{2-1/18})$ for general $\Delta$ \cite{ChanCGKLZ25} \\  
  $O^*(n^{2-1/9})$ for $\Delta=O(1)$ \cite{ChanCGKLZ25} }  \\
  &  &  & \cellcolor{Highlight} $\OO(n^{4/3})$ for $\Delta=2$ (Thm.~\ref{thm:udgdiam2}) \\
\hline
\hline
\multirow{1}{*}{\shortstack{Axis-aligned\\squares}} 
  & & & \shortstack{ $O(n\log n)$ for $\Delta=2$ \cite{Bringmann2022-me}\\
       $O^*(n^{7/4})$ for $\Delta=O^*(1)$ \cite{duraj2023better}\\
       $\OO(n^{2-1/8})$ for general $\Delta$ \cite{ChanCGKLZ25} } \\
  & unit squares & & \cellcolor{Highlight} $O^*(n)$ for constant $\Delta$ (Thm.~\ref{thm:unitsq}) \\
    \cline{2-4}
  & general & & $O^*(n^{2-1/12})$ for general $\Delta$ \cite{ChanCGKLZ25} \\
\hline
\hline
\multirow{2}{*}{\shortstack{Triangles\\ in 2D}} 
  & \shortstack{cong.\ \& equi.} &  $\Omega^*(n^2)$ for $\Delta=3$ {\tiny (OV)} \cite{Bringmann2022-me} & 
  \\
  \cline{2-4}
  & fat &  \cellcolor{Highlight} $\Omega^*(n^2)$ for $\Delta=2$ {\tiny (3H6)} (Thm.~\ref{thm:lb_fattriangle}) & 
  \\
\hline
\hline
\multirow{7}{*}{\shortstack{Line\\ segments}}
  & unit-length &  $\Omega^*(n^2)$ for $\Delta=3$ {\tiny (OV)} \cite{Bringmann2022-me} & \\
  \cline{2-4}
  & general & \cellcolor{Highlight} $\Omega^*(n^2)$ for $\Delta=2$ {\tiny (combK4)} (Thm.~\ref{thm:segment_lower}) & 
  \\
  \cline{2-4}
  & \shortstack{axis-aligned \\ w/ degeneracy} & $\Omega^*(n^2)$ for general $\Delta$ {\tiny (OV)} \cite{Bringmann2022-me} & 
    \cellcolor{Highlight} \shortstack{$n^{2-1/O(2^{\Delta})}$ for any constant $\Delta$ \\(Thm.~\ref{thm:h-slope-line-seg}) }
  \\
  \cline{2-4}
  & \shortstack{axis-aligned \\ non-degenerate} & & \cellcolor{Highlight} \shortstack{$\OO(n^{2-1/32})$ for general $\Delta$ \\ (Thm.~\ref{thm:gen-pos-line-seg}) }
  \\
  \cline{2-4}
  & \shortstack{$h$-slopes 
  \\ w/wo degen.} & 
  \cellcolor{Highlight} \shortstack{$\Omega^*(n^2)$ for $h=3$ \&\\ general $\Delta$ {\tiny (OV)} (Thm.~\ref{thm:segment_3lope_lower}) } & 
  \cellcolor{Highlight} \shortstack{$n^{2-1/O(h^{\Delta+1})}$ for constant $h$ \\ \& constant $\Delta$ (Thm.~\ref{thm:h-slope-line-seg}) }
  \\
\hline
\hline
\multirow{2}{*}{Strings} 
  & $O(1)$ bends  & \cellcolor{Highlight} $\Omega^*(n^2)$ for $\Delta=2$ {\tiny (3H6)} (Thm.~\ref{thm:lb_thintriangle}) &   
  \\
  \cline{2-4}
  & $O(\log n)$ bends & \cellcolor{Highlight} $\Omega^*(n^2)$ for $\Delta=1$  {\tiny (OV)} (Thm.~\ref{thm:string_lower}) & 
  \\
\hline
\hline
\end{NiceTabular}
\caption{
 Previous and new time bounds for deciding whether an intersection graph of geometric objects has diameter at most $\Delta$.  New results are highlighted in yellow.  Conditional lower bounds marked ``(OV)'', ``(3H6)'', and ``(combK4)'' assume the Orthogonal Vectors hypothesis, the 3-uniform 6-hyperclique hypothesis, and the combinatorial 4-clique hypothesis respectively, where the latter is for combinatorial algorithms (all upper bounds are obtained from combinatorial algorithms). See \Cref{appendix:hypothesis} for definitions. 
}
\label{table:results-2D}
\end{table}

\subsection{Our Contributions}\label{subec:contribution}

We focus on the intersection graphs of objects in 2D; see \Cref{table:results-2D} for a complete catalog. Our first set of results concerns line segments. Previous results about line segments~\cite{Bringmann2022-me} are overwhelmingly negative: for unit line segments, \Diameter-3 is (truly subquadratic) hard, and for axis-aligned unit segments, \Diameter-$\Delta$ is also hard when $\Delta = \Omega(\log n)$. For axis-aligned unit segments, the lower bound construction requires that there is an overlap between horizontal segments; in other words, the instance is \EMPH{degenerate}. Both the requirement $\Delta = \Omega(\log n)$ and the degeneracy seem like artifacts of the proof technique: better lower bounds seem achievable with a more refined technique. Surprisingly perhaps, we show that these ``artifacts'' are necessary by giving truly subquadratic-time algorithms otherwise. First, we provide a truly subquadratic algorithm for any diameter value on non-degenerate instances. Thus, degeneracy is a source of hardness for diameter computation.  Second, we provide a truly subquadratic algorithm for \Diameter-$\Delta$ for any \emph{constant} $\Delta$, even when the instances are degenerate. It implies that the diameter range is another source of hardness. Our second result holds even for a more general case of \EMPH{$h$-slope line segments}, where the slope of each segment is one of  $h$ distinct values. (Axis-parallel line segments are $2$-slope.) Finally, on the lower bound side,  we strengthen the (truly subquadratic) lower bound for line segments by Bringmann \etal~\cite{Bringmann2022-me} from \Diameter-3   under the OV hypothesis to \Diameter-2   under the \EMPH{combinatorial 4-clique (combK4)} hypothesis (\Cref{def:combK4}).
We also give a separation between 2-slope and $3$-slope non-degenerate segments: for 2-slope, we can get truly subquadratic time for any diameter value, whereas for 3-slope, the problem becomes hard for diameter $\Omega(\log n)$.

\begin{theorem}\label{thm:segments-main} Let $G$ be the intersection graph of $n$ line segments.
\begin{enumerate}
    \item \ul{Truly subquadratic algorithms:} 
    \begin{itemize}
        \item   \Diameter can be solved in $\OO(n^{2-1/32})$ time for non-degen.\ axis-aligned line~segments.
        \item   \Diameter-$\Delta$ can be solved in  $n^{2-1/O(h^{\Delta+1})}$ time for $h$-slope line segments. Thus, the running time is truly subquadratic for any constant $h$ and $\Delta$. 
    \end{itemize}
    \item  \ul{Lower bounds:}    there is no $O(n^{2-\eps})$ time algorithm for any constant $\eps \in (0,1)$ for:
    \begin{itemize}
        \item \Diameter-2 for general line segments  under the CombK4 hypothesis.
        \item  \Diameter-$\Omega(\log n)$ for 3-slope non-degenerate line segments under the OV hypothesis.
    \end{itemize}
   \end{enumerate}
\end{theorem}

To obtain the algorithms in \Cref{thm:segments-main}, we apply the framework of Chan \etal~\cite{ChanCGKLZ25}: (1) bounding the VC-dimension of the set system $(V, \{N^{r}[v]\}_{v\in V, r\geq 0})$, where $N^{r}[v]$ is the $r$-neighborhood of $v$, and (2) designing a geometric data structure for the objects. For line segments, the geometric data structure follows from standard techniques. The main difficulty is to bound the VC-dimension. We do so by decomposing $(V, \{N^{r}[v]\}_{v\in V, r\geq 0})$ into so-called \EMPH{types} where each type induces a subset of set systems whose corresponding shortest paths share the same ``color pattern''. We use the color pattern on top of the standard crossing argument to bound the VC-dimension of each subset. Our technique naturally generalizes to string graphs, giving the following theorem.
See \Cref{sec:segs} for the details.

\begin{theorem}\label{thm:string}  Let $G$ be the intersection graph of $n$ strings on the plane. 
\begin{enumerate}
        \item If $G$ is bipartite, then the VC-dimension of $(V, \{N^{r}[v]\}_{v\in V, r\geq 0})$ is at most 8.
        \item  If $G$ has diameter $\Delta$ and chromatic number $\chi$, then the VC-dimension of $(V, \{N^{r}[v]\}_{v\in V, r\geq 0})$ is $O(\chi^{\Delta+1})$.
  \end{enumerate}    
\end{theorem}

Our lower bound for \Diameter-2 for general line segments in \Cref{thm:segments-main} is inspired by the lower bound for \Diameter-2 for axis-parallel hypercubes in $\reals^{12}$ from the hyperclique hypothesis by Bringmann \etal~\cite{Bringmann2022-me}.  
It is somewhat surprising that their proof approach can be carried out
in dimension as low as 2 for line segments---a key challenge in such a reduction
is avoiding unintended crossings of objects, which is tougher to do in 2D;
see \Cref{sec:lb:seg:diam2}.

The VC-dimension underpins most recent truly subquadratic diameter algorithms. A major open question is to design an algorithm with nearly linear running time for the truly subquadratic easy cases. The VC-dimension technique does not seem suited for this fast running time. Only \Diameter-2 for axis-aligned unit squares~\cite{Bringmann2022-me} is known to admit $\OO(n)$ time. 
Here we give an almost linear  ($O^*(n)$) time algorithm\footnote{$O^*(\cdot)$ notation hides a subpolynomial factor $n^{o(1)}$.} for unit-square graphs of \emph{any constant diameter}.
Our algorithm is far more involved than the previous algorithm for diameter 2 
(which did not work even for diameter~3), and uses
a novel $n^{o(1)}$-way divide-and-conquer approach based on the coordinate values modulo~1. 
We identify a key special case 
that can be reduced to orthogonal range searching,
not in dimension~2 but in a dimension that grows as a function of the diameter!  Afterwards,
we recursively reduce to this key special case. 
See \Cref{sec:unitsq:diam2} for the details. 
We are not aware of too many low-dimensional geometric problems that are solved using orthogonal range searching in a much higher dimension 
as a function of some parameter (one example is Cabello and Knauer's algorithms for distance problems 
on bounded-treewidth graphs~\cite{CK09}, but our algorithm is quite different).

Another simple, yet highly non-trivial, problem is \Diameter-2 for unit disks; orthogonal range searching does not seem to help here. The best running time, based on VC-dimension, is $\OO(n^{2-1/18})$~\cite{ChanCGKLZ25}. Here, we improve the running time for \Diameter-2 for unit disks to $\OO(n^{4/3})$. Notably, we do not use VC-dimension.
Instead, we exploit the fact that 1-neighborhoods behave like \emph{pseudodisks}, and adapt known range searching techniques for pseudodisks.
Note that our pseudodisks have non-constant complexity and can only be represented implicitly (similar ideas were used before in Agarwal, Sharir, and Welzl's algorithm for the discrete 2-center problem~\cite{ASW98}). 
See \Cref{sec:diam2} for the details.
The bound $n^{4/3}$ is natural for many problems about unit disks (related to Hopcroft's problem); for example, the current best algorithm for counting the number of edges in a unit disk graph requires $O(n^{4/3})$ time~\cite{ChanZ24}.

\begin{restatable}{theorem}{diameterunitdis}\label{thm:diameter2-unitdis} Let $G$ be an intersection graph with $n$ vertices.
 \begin{itemize}
        \item  
        \Diameter-$\Delta$ can be solved in  $O^*(n)$ time when vertices of $G$ are axis-aligned unit squares. 
        \item   $\Diameter$-2 can be solved in $\OO(n^{4/3})$ time when vertices of $G$ are unit disks.
    \end{itemize}
\end{restatable}

\medskip
Finally, we show new lower bounds for triangles and strings. Since line segments are strings, our lower bounds for line segments also hold for strings. For triangles, we show that \Diameter-2 is subquadratically hard even when the triangles are fat; the previous lower bound is only for \Diameter-3 and non-fat triangles~\cite{Bringmann2022-me}. Here we follow the same idea in our lower bound for \Diameter-2 of line segments in \Cref{thm:segments-main}. The key challenge is handling the unintended crossings of triangles in 2D. The fatness requirement makes the overall argument more delicate. We then obtain the lower bound for \Diameter-2 of strings with $O(1)$ complexity as a simple corollary; details are given in \Cref{sec:triangle}. For strings of  $O(\log n)$ complexity, we show that  \Diameter-1  is subquadratically hard. Observe that \Diameter-1 problem is equivalent to deciding if the intersection graph of a given set of objects is a clique.  For objects that are semialgebraic sets of constant description complexity,
\Diameter-1 can be solved in truly subquadratic time by known data structures for intersection searching~\cite{AgarwalE99}. Our lower bound in \Cref{thm:string-rectangle} shows that the problem is hard for strings of  $O(\log n)$ complexity by a reduction from OV. The idea is to create a polygonal chain of $d$ vertices from each vector of dimension $d$, such that there are two orthogonal vectors if and only if there are two non-intersecting chains; details are given in \Cref{sec:diameter1-string}. 


\begin{theorem}\label{thm:string-rectangle}  There is no $O(n^{2-\eps})$ time algorithm for constant $\eps \in (0,1)$ for:
\begin{itemize}
        \item \Diameter-$2$ of fat triangles or strings of $O(1)$ complexity under the 3-uniform 6-hyperclique hypothesis.
        \item   \Diameter-$1$ of strings of $O(\log n)$ complexity under OV hypothesis.
\end{itemize}
\end{theorem}

\section{Preliminaries}
\paragraph{Graph preliminaries.}
Throughout this paper, we will let $G = (V,E)$ be a geometric intersection graph of $n = |V|$ geometric objects $\cO$. We will assume that we are given the geometric objects $\cO$. In this paper, we mostly consider connected objects that have $O(1)$ complexity, unless noted otherwise. We use $d(u,v)$ to denote the distance between two vertices $u,v\in V$ in~$G$.
We denote the 1-hop \EMPH{neighborhood} of a vertex $v\in V$ by $N[v]$, and the \EMPH{$r$-neighborhood ball} of $v$ by $N^r[v]:= \{u\in V : d(u,v) \le r\}$. For any two vertices $u$ and $v$, we denote by $\pi_G(u,v)$ the shortest path between them and  $|\pi_G(u,v)|$ the number of edges of the path.

\paragraph{VC-dimension.}
Given a set system  $(U, \mathcal{F})$ with a ground set $U$ and a family $\mathcal{F}$ of subsets of $U$, its \EMPH{VC-dimension} is the cardinality of the largest $S\subseteq U$ such that $S$ is \emph{shattered} by $\mathcal{F}$---%
for every $S'\subseteq S$, there is some $X\in \mathcal{F}$ such that $X\cap S = S'$. We denote this VC-dimension by $\vcdim(U, \mathcal{F})$. 
We say that a graph $G$ has \EMPH{distance VC-dimension} at most $d$ if the set system of neighborhood balls $(V_G, \{N^r[v]\}_{v\in V, r \geq 0})$ has VC dimension at most $d$.

Let $\mathcal{A}$ and $\mathcal{B}$ be two families of subsets of the same ground set $U$. We define $\mathcal{A}\cup \mathcal{B} \coloneqq \{A\cup B: A\in \mathcal{A}, B\in \mathcal{B}\}$. It is well known that:
\begin{lemma}[Vapnik~\cite{Vapnik1998}]\label{lm:vc-union} $\vcdim(U,\mathcal{A}\cup \mathcal{B}) \leq \vcdim(U,\mathcal{A})  + \vcdim(U,\mathcal{B})$.
\end{lemma}

\paragraph{Fine-grained complexity hypotheses.}\label{appendix:hypothesis}

Recent development of fine-grained complexity has identified a few hypotheses that we include here for completeness. We use them to prove our lower bounds. 

\begin{definition}[Orthogonal Vectors (OV) hypothesis]\label{def:OV}
    Given sets $A, B$ of $n$ vectors in $\{0,1\}^d$ with $d=\omega(\log n)$, deciding whether there exists an orthogonal pair $(a,b)\in A\times B$ requires $n^{2-o(1)}$ time.
\end{definition}
The OV hypothesis is implied~\cite{Williams2005-nr} by the Strong Exponential Time Hypothesis~\cite{Impagliazzo2001-aw}.

\begin{definition}[3-uniform 6-hyperclique (3H6) hypothesis]\label{def:3H6}
Given a $6$-partite $3$-uniform hypergraph $G=(V,E)$ where $V$ is the disjoint union of vertex set $V^{(1)}, \ldots, V^{(6)}$, each containing $n$ vertices, and $E\subseteq {V \choose 3}$ such that each edge connects three vertices from different vertex sets. The problem is to decide whether there are $6$ vertices $S=\{v_1,v_2,\ldots,v_6\}$ with $v_i\in V^{(i)}$, $i=1,\ldots,6$, forming a $6$-clique, i.e., $\{v_i,v_j,v_k\}\in E$ for all $\{i,j,k\}\in {S\choose 3}$.
The 3H6 hypothesis says that the problem requires $n^{6-o(1)}$ time. 
\end{definition}

More information about the hyperclique hypothesis can be found in~\cite{Lincoln2018-ei}.

\begin{definition}[Combinatorial 4-clique (combK4) hypothesis]\label{def:combK4}
    Any combinatorial algorithm detecting whether a graph of $n$ vertices contains a $4$-clique requires $n^{4-o(1)}$ time.
\end{definition}

More information about the combinatorial $4$-clique problem and connections to other hypotheses can be found in~\cite{Chan2008-oe,Abboud2015-oc}.

\section{Line Segments and String Graphs}\label{sec:segs}

A \EMPH{string graph} is an intersection graph of curves in the plane. Each curve is called a \EMPH{string}. Note that intersection graphs of line segments are string graphs.
In this section, 
we will devote to proving VC-dimension bounds for string graphs that will give subquadratic diameter algorithms for line-segment intersection graphs. Once we have upper bounds for VC-dimension, we can apply (variations of) the algorithmic framework of~\cite{ChanCGKLZ25} to obtain subquadratic time algorithms.

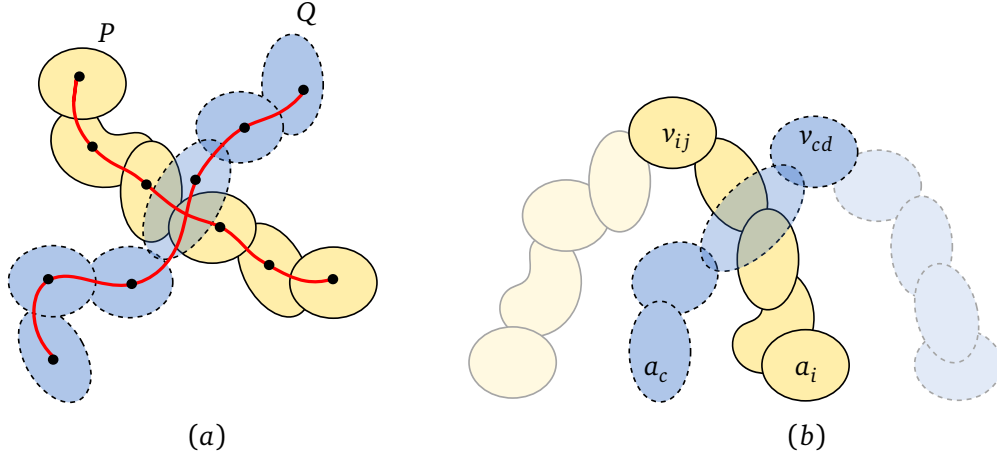
\begin{figure}
    \centering
    \begin{tikzpicture}
    \node[anchor=south west, inner sep=0] (image) at (0,0) {
       \includegraphics[width=0.8\textwidth]{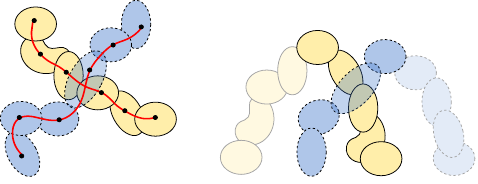}
    };

    \begin{scope}[x={(image.south east)},y={(image.north west)}]
        \node[black] at (0.1, 1.0) {$P$};
        \node[black] at (0.3, 1.05) {$Q$};
         \node[black] at (0.2, -0.05) {$(a)$};
         \node[black] at (0.8, -0.05) {$(b)$};
         \node[black] at (0.67, 0.72) {$v_{ij}$};
         \node[black] at (0.8, 0.12) {$a_{i}$};
         \node[black] at (0.81, 0.72) {$v_{cd}$};
         \node[black] at (0.65, 0.12) {$a_{c}$};
    \end{scope}
\end{tikzpicture}
    \caption{(a) Two (vertex-disjoint) paths $P$ and $Q$ cross on the plane; their canonical (red) curves cross. (b) Two paths $\pi_G(v_{ij}, a_i)$ (yellow)  and $\pi_G(v_{cd}, a_c)$ (blue) cross on the plane; in this figure, $|\pi_G(v_{ij},a_c)|\leq |\pi_G(v_{ij},a_i)|$.}
    \label{fig:canonical-curve}
\end{figure}

Our key technical contribution in this section is a new color-based argument to bound the VC dimension. To set the stage, we need to introduce the notion of canonical curves. Let $G$ be the intersection graph of (connected) objects. For each object $o_x$ corresponding to a vertex $x\in V$, we pick a point $p_x\in o_x$ as a representation of $o_x$.   For a path $P$  between two vertices $x$ and $y$ in $G$, there is a planar curve, denoted by $\sigma_P$, that connects $p_x\in o_x$ and $p_y\in o_y$, and traces the objects corresponding to the vertices in $P$. We call $\sigma_P$ the \EMPH{canonical curve} of $P$. For a pair of vertices $x$ and $y$, we define their \EMPH{canonical curve $\sigma(x,y)$} to be the canonical curve of the shortest path in $G$ between $x$ and $y$.  (We will fix one arbitrary shortest path per pair to define the canonical curve.)  We say that two paths $P$ and $Q$ of $G$ \EMPH{cross on the plane} if they do not share any endpoint and their canonical curves intersect in $\R^2$. Note that $P$ and $Q$ might be vertex-disjoint; see \Cref{fig:canonical-curve}(a).
  
Our starting point is the \EMPH{$K_5$-avoidance argument} by Chang, Gao, and Le~\cite{CGL24}  for showing that the distance VC-dimension of any intersection graph of pseudo-disks is at most 4, which we briefly review here.  The proof is by contradiction: suppose that 5 vertices $\{a_1,a_2,\ldots, a_5\}$ are shattered, then, for every two points $a_{i}$ and $a_{j}$ where $i\not= j \in [5]$, there is a (distinct) vertex $v_{ij}$ whose neighborhood contains $a_i$ and $a_j$ only. The two canonical curves $\sigma(a_i,v_{ij})\circ \sigma(v_{ij},a_j)$ form a curve between $p_{a_i}$ and $p_{a_j}$ on the plane. The set of 5 points $\{p_{a_i}\}_{i\in [5]}$ and 10 curves  $\{\sigma(a_i,v_{ij})\circ \sigma(v_{ij},a_j)\}_{i\not=j \in [5]}$ between these points form a  drawing of $K_5$ on the plane. Then, by the Hanani–Tutte theorem, there exist two curves that do not share endpoints and cross (an odd number of times). Consequently, there exists $v_{ij}$ and $v_{cd}$ for  four distinct indices $i,j,c,d \in[5]$ such that  $\pi_G(v_{ij}, a_i)$  and $\pi_G(v_{cd}, a_c)$ cross on the plane.  Using the fact that objects are pseudo-disks, they conclude that either (i) $|\pi_G(v_{cd}, a_i)| \leq |\pi_G(v_{cd}, a_c)|$ or (ii) $|\pi_G(v_{ij}, a_c)| \leq |\pi_G(v_{ij}, a_i)|$; see \Cref{fig:canonical-curve}(b). In the former case, the neighborhood of $v_{cd}$ contains $a_i$, contradicting that $v_{cd}$ only contains $a_c$ and $a_d$; in the latter case, the neighborhood of $v_{ij}$ contains $a_c$, contradicting that $v_{ij}$ only contains $a_i$ and $a_j$. The last part of the argument depends crucially on the fact that objects are pseudo-disks, and therefore, does not apply to strings and line segments (which are not pseudo-disks).

Here, we devise a new argument based on vertex coloring on top of the $K_5$-avoidance technique. Our idea is best illustrated in the case of bipartite string graphs $G = (A\cup B, E)$ (in \Cref{subsec:bipartite-string}). Specifically, each vertex can be colored as $A$ or $B$, depending on which side it belongs to.  Then we can view each canonical curve $\sigma(a_i,v_{ij})$ as a sequence $S_{ij}$ of alternating colors: $S_{ij} = ABAB\ldots $. Following the setup of the $K_5$-avoidance argument, when $\pi_G(v_{ij}, a_i)$  and $\pi_G(v_{cd}, a_c)$ cross on the plane, our key idea is to align the corresponding color sequences $S_{ij}$ and $S_{cd}$ at the crossing objects, which have \EMPH{different colors}. Then, by a simple case analysis, we can still show that either  (i) $|\pi_G(v_{cd}, a_i)| \leq |\pi_G(v_{cd}, a_c)|$ or (ii) $|\pi_G(v_{ij}, a_c)| \leq |\pi_G(v_{ij}, a_i)|$ as in the case of pseudo-disks and get a contradiction. We note that the color alignment argument holds only when there are two colors (due to bipartiteness); it does not hold when there are three or more colors. Indeed, for the case of three colors, we have a quadratic conditional lower bound by \Cref{thm:segments-main}, and the VC-dimension is unbounded. 

For string graphs of bounded chromatic number and bounded diameter, we decompose the set system $(V, \{N^{r}[v]\}_{v\in V, r\geq 0})$  into set systems of \EMPH{different types}, each type corresponds to a fixed sequence of (vertex) colors. Then, for each sequence of colors $S$, we define a notion of ball $B_S(v)$ that contains all vertices reachable from $v$ by paths whose vertex colors form a subsequence of $S$; the goal is then to bound the VC dimension of  $(V, \{B_S(v)\}_{v\in V})$. Our key insight is that, by restricting the color sequences in the definition of $B_S(v)$  to be subsequences of $S$, we could apply our color alignment argument again.  The case analysis is more delicate, but the underlying idea remains the same as in bipartite string graphs. The proof is given in \Cref{subsec:bounded-chromatic}. 

Finally, in \Cref{subsection:slope}, we adapt the argument for string graphs of bounded chromatic number and bounded diameter to (possibly degenerate) line segments with few slopes. Here, the new issue is that the colors are not proper: the endpoints of an edge could share the same color. We get around this issue by exploiting the fact that objects are line segments. Roughly speaking, if two segments $s$ and $\Tilde{s}$  of the same slope, say horizontal, intersect, then one of the neighbors of $s$ will intersect $\Tilde{s}$ (or vice versa). This means we can align the color of a neighbor of $s$ with that of $\Tilde{s}$, and hence the color alignment argument can be applied again. 

\subsection{Bipartite String Graphs}\label{subsec:bipartite-string}
\begin{theorem}\label{thm:string-bipartite} Every bipartite string graph $G = (A\cup B,E)$ has distance VC-dimension at most $8$.
\end{theorem} 

\begin{proof} 
Suppose that at least $9$ vertices are shattered by the set system, then at least $5$ vertices in one side, say $A$, are shattered. Let these be $X_A \coloneqq \{a_1,a_2,\ldots, a_5\}$. Let \EMPH{$v_{ij}$} and \EMPH{$r_{ij}$} be such that $X_A\cap N^{r_{ij}}[v_{ij}] = \{a_i,a_j\}$. We call $r_{ij}$ the the radius \emph{associated with} $v_{ij}$.

For every string $s_x$ corresponding to a vertex $x\in V$, we pick a point $p_x\in s_x$, as a representation of $s_x$.  
Recall that for every two vertices $x$ and $y$ and a shortest path $\pi_{G}(x,y)$ between them in graph $G$, the canonical curve $\sigma(x,y)$ is a planar curve tracing the strings corresponding to vertices in $\pi_{G}(x,y)$ that connects $x\in s_x$ and $y\in s_y$. 

Write the vertices of the canonical curve from $a_i$ to $v_{ij}$ as a sequence $S_{ij}$ of alternating colors of length $\ell_i\le r_{ij}+1$: $S_{ij} = ABAB\ldots$ where each color corresponds to which set (either $A$ or $B$) the vertices belong. The initial color of $S_{ij}$ is $A$ since $a_i\in A$.  We'll use the notation \EMPH{$S_{ij}[k]$} to denote the $k$th color in the sequence $S_{ij}$ and the notation
\EMPH{$S_{ij}[k_1 \dots k_2]$} to denote the contiguous subsequence of $S_{ij}$ between the indices $k_1$ and $k_2$.
  
To get a contradiction, we follow the $K_5$-avoidance argument using the Hanani–Tutte theorem. Specifically, observe that $X_A$ and the set of curves $\{\sigma(a_i,v_{ij})\circ \sigma(v_{ij},a_j)\}_{1\leq i,j\leq 5}$ form a drawing of $K_5$ on the plane. 
Then by the Hanani–Tutte theorem, there exist two curves not sharing endpoints that cross an odd number of times, i.e., the curves cross at least once. 
Suppose that the crossing pair of canonical curves were $\sigma(a_i, v_{ij})$ and $\sigma(a_c, v_{cd})$ for four distinct indices $i,j,c,d \in[5]$.
We will show that the following two facts result in a contradiction.
\begin{enumerate}
    \item There are two vertices $x_i$ corresponding to a string on the canonical curve $\sigma(a_i, v_{ij})$ and $x_c$ corresponding to a string on $\sigma(a_c, v_{cd})$ where either $x_i = x_c$ or $x_i$ crosses $x_c$ (i.e. $(x_i,x_c) \in E$).
    \item $a_c \notin N^{r_{ij}}[v_{ij}]$ and 
    $a_i \notin N^{r_{cd}}[v_{cd}]$.
\end{enumerate}

Suppose that $x_i$ was the $k_i$th vertex out of $\ell_i$ total vertices on the canonical curve $\sigma(a_i, v_{ij})$
and that $x_c$ was the $k_c$th vertex out of $\ell_c$ total vertices on the canonical curve $\sigma(a_c, v_{cd})$. Recall that as $x_c \in N^{r_{cd}}[v_{cd}]$, we must have $\ell_c \le r_{cd}+1$.
Without loss of generality, assume that $k_i \le k_c$. 

First we consider the case where $x_i = x_c$. Then there must be a path from $a_i$ to $v_{cd}$ corresponding to the concatenated sequence $S_{ij}[1\dots k_i] \circ S_{cd}[(k_c+1)\dots \ell_c]$. However the length of this sequence is $k_i+(\ell_c-k_c) \le \ell_c \le r_{cd}+1$, which means that $a_i \in N^{r_{cd}}[v_{cd}]$ which is a contradiction.

Next we consider the case that $x_i$ and $x_c$ cross. Note that exactly one of $x_i$ and $x_c$ is in $A$ while the other is in $B$. In particular, this means that $k_i \neq k_c$, since both $S_{ij}$ and $S_{cd}$ are alternating sequences beginning with A, so $k_i < k_c$. As both are integers, $k_i+1 \le k_c$. 
Now this means that the path corresponding with the sequence $S_{ij}[1\dots k_i] \circ S_{cd}[k_c\dots \ell_c]$ has length at most $k_i + (\ell_c-k_c + 1) \le \ell_c \le r_{cd}+1$ which again means that $a_i \in N^{r_{cd}}[v_{cd}]$, a contradiction.
\end{proof}

In \Cref{ssec:diameter-NDAALS} we show that this VC-dimension bound can be easily used in the framework of \cite{ChanCGKLZ25} to obtain a subquadratic time algorithm.
\begin{restatable}{theorem}{GenPosLineSeg}
\label{thm:gen-pos-line-seg}
The diameter of axis-aligned segments in general position can be computed in $\OO(n^{2-1/32})$ time.
\end{restatable}

This sharply contrasts with the conditional lower bound by Bringmann \etal~\cite{Bringmann2022-me} who showed that computing the diameter of an intersection graph of axis-aligned segments that are \emph{not} in general position requires $\Omega(n^{2-\eps})$ time for any $\eps > 0$.

\subsection{String Graphs of Bounded Diameter and Chromatic Number}\label{subsec:bounded-chromatic}

In this section, we prove the following VC-dimension bound for string graphs with bounded chromatic number and bounded diameter.
\begin{theorem}\label{thm:string-bounded} Let $G = (V,E)$ be a string graph with chromatic number $\chi$ and diameter $\Delta$. Then $(V, \{N^{r}[v]\}_{v\in V, r\geq 0})$ has VC-dimension at most $O(\chi^{\Delta+1})$.
\end{theorem} 

The key insight into proving \Cref{thm:string-bounded} is to decompose  $(V, \{N^{r}[v]\}_{v\in V, r\geq 0})$  into set systems of \EMPH{different types}, and show that each type has bounded VC-dimension. First, we introduce the types. 
Let \EMPH{$\mathcal{S}$} be the set of sequences of length at most $\Delta+1$ where each element of a sequence is a color in $[\chi]$. Observe that $|\mathcal{S}|  = O(\chi^{\Delta+1})$.

We will use the notation $P(v\rightarrow u)$ to denote a directed path obtained from an undirected path $P(v,u)$ by directing all the edges away from $v$. If $a$ and $b$ are two vertices in $P(v\rightarrow u)$ such that $b$ is reachable from $a$ on the path, then we denote by $P(v\rightarrow u)[a,b]$ the subpath from $a$ to $b$. We consider a proper coloring of the vertices in $G$ with at most $\chi$ different colors. The colors of the vertices in $P(v\rightarrow u)$ induce a sequence of colors, say $S$. We say that  $P(v\rightarrow u)$ is an \EMPH{$S$-rainbow path} from $v$ to $u$. If a sequence of colors $S'$ is a (not necessarily contiguous) subsequence of $S$, we write \EMPH{$S' \sqsubseteq S$}.  For every sequence $S\in \mathcal{S}$, let
\begin{equation}
    \EMPH{$B_S(v)$} \coloneqq \{ u\in V: \text{ $\exists S'\sqsubseteq S$ and an $S'$-rainbow path $P(v\rightarrow u)$} \}
\end{equation}

We show that these carefully defined balls have bounded VC-dimension using a generalization of the bipartite argument.
\begin{lemma}\label{lm:VCbound-type} $\vcdim(V, \{B_{S}(v)\}_{v\in V})\leq 4$ for any sequence of colors $S$. 
\end{lemma}

\begin{proof} We follow the setup  (and proof technique) of \Cref{thm:string-bipartite}. 
In particular, $X = \{a_1,\ldots, a_5\}$ is the set of vertices shattered by $\set{B_S(v)}_{v\in V}$, and $v_{ij}$ is such that $X\cap B_{S}(v_{ij}) = \{a_i,a_j\}$. Then, there exist four distinct indices $i,j,c,d\in [5]$ inducing 
an $S_i$-rainbow path $P_i$ with canonical curve $\sigma(v_{ij}, a_i)$ and 
an $S_c$-rainbow path $P_c$ with canonical curve $\sigma(v_{cd}, a_c)$ which cross. 
Note that $S_i \sqsubseteq S$ and $S_c \sqsubseteq S$. 
We show that the following facts result in a contradiction.
  \begin{enumerate}
    \item There are two vertices $x_i$ corresponding to a string on the canonical curve $\sigma(v_{ij}, a_i)$ and $x_c$ corresponding to a string on $\sigma(v_{cd},a_c)$ where either $x_i = x_c$ or $x_i$ intersects $x_c$ (i.e. $(x_i,x_c) \in E$).
      \item  $a_i\not\in B_{S}(v_{cd})$ and $a_c\not\in  B_{S}(v_{ij})$.
  \end{enumerate}

As $S_i$ is a subsequence of $S$, there exist some mapping $I_S^i(k)$ for $k\in \{1,\dots, |S_i|\}$ such that if $I^i_S(k) = h$ then $S_i[k] = S[h]$, and for all $\ell_1 < \ell_2$, we have that $I_S^i(\ell_1) < I_S^i(\ell_2)$ (this is by the formal definition of subsequence). 
We call $I_S^i(k)$ the \EMPH{$S$-index} of $S_i[k]$. We similarly define the $I_S^c(k)$ to be the $S$-index of $S_c[k]$. 
Let $k_i$ denote the index of $x_i$ in the path $P_i$ and $k_c$ for the index of $x_c$ in $P_c$. Without loss of generality, assume that $I^i_S(k_i) \le I_S^c[k_c]$.

Now consider if $x_i = x_c$. There must exist a path from $v_{ij}$ to $x_i=x_c$ to $a_c$ with the sequence $S' =S_i[1\dots k_i]\circ S_c[(k_c+1)\dots|S_c|]$. 
Furthermore, $S_i[1\dots k_i]\sqsubseteq S[1\dots I^i_S(k_i)]$ and $S_c[(k_c+1) \dots |S_c|] \sqsubseteq S[(I_S^c[k_c]+1)\dots|S|]$.
Since $I^i_S(k_i) \le I_S^c[k_c]$, then it follows that 
$S'\sqsubseteq S$ which is a contradiction as this would mean $a_c \in B_S(v_{ij})$.

On the other hand suppose there is an edge between $x_i$ and $x_c$. 
Note that as $x_i$ and $x_c$ have different colors, we must have $I^i_S(k_i)\neq I_S^c(k_c)$, meaning that $I^i_S(k_i)< I_S^c(k_c)$ and as $S$-indices take on integer values, $I^i_S(k_i)+1\le I_S^c(k_c)$.
There must exist a path from $v_{ij}$ to $x_i=x_c$ to $a_c$ with the sequence $S''=S_i[1\dots k_i]\circ S_c[k_c\dots|S_c|]$. 
As $S_i[1\dots k_i]\sqsubseteq S[1\dots I^i_S(k_i)]$ and $S_c[k_c \dots |S_c|] \sqsubseteq S[I_S^c[k_c]\dots|S|]\sqsubseteq S[(I_S^i[k_i]+1)\dots|S|]$, this shows that $S''\sqsubseteq S$. 
This means $a_c\in B_S(v_{ij})$ which is a contradiction.
\end{proof}

\begin{proof}[Proof of \Cref{thm:string-bounded}] Since the diameter of $G$ is $\Delta$, every shortest path $P(u\rightarrow v)$ from $u$ to $v$ induces a sequence of at most $\Delta+1$ colors.
Recall that $\mathcal{S}$ is the set of all sequences of colors of length at most $\Delta+1$, so $|\mathcal{S}| \leq O(\chi^{\Delta+1})$.
Let $\mathcal{B}_S = \{B_{S}(v)\}_{v\in V}$ and $\mathcal{B}_G = \{N^{r}[v]\}_{v\in V, r\geq 0}$. Observe that $\mathcal{B}_G\subseteq\bigcup_{S\in \mathcal{S}}\mathcal{B}_S$, and hence
\begin{equation*}
 \vcdim(V,\mathcal{B}_G) \overset{\text{Lem.~\ref{lm:vc-union}}}{\leq}  \sum_{S\in \mathcal{S}}\vcdim(V,\mathcal{B}_S) \overset{\text{Lem.~\ref{lm:VCbound-type}}}{\leq}  4|\cS| = O(\chi^{\Delta+1}). 
\end{equation*}
\end{proof}

\subsection{Intersection Graphs of Line Segments with Few Slopes}\label{subsection:slope}

Recall that we say that a collection of line segments is \EMPH{$h$-slope} if the slopes of the line segments have $h$ distinct values. To show the following theorem,
we will again use the Hanani-Tutte theorem, but now we must allow overlaps for the curves. We say that two curves cross if they cross in the traditional sense after some sufficiently small local perturbations, as in~\cite{FulekK18}.

\begin{theorem}\label{thm:segments} Let $G = (V,E)$ be an intersection graph of $h$-slope line segments with diameter  $\Delta$. Then $(V, \{N^{r}[v]\}_{v\in V, r\geq 0})$ has VC-dimension at most $O(h^{\Delta+1})$.
\end{theorem} 

\begin{proof}
It is tempting to directly apply \Cref{thm:string-bounded} to prove \Cref{thm:segments}, but a set of $h$-slope line segments may be degenerate, so two line segments of the same slope can intersect, and therefore, the graph is not $h$-colorable. However, we can still create a (non-proper) coloring of the geometric intersection graph $G$ where each line segment is colored by its slope, and define $S$-rainbow paths for sequences of colors $S$ as before. The only issue occurs in the proof of \Cref{lm:VCbound-type}.
If two canonical paths considered share any vertices we can apply the argument in \Cref{lm:VCbound-type}.

It suffices to show that we can always find two adjacent vertices that correspond to different $S$-indices.
Unfortunately, the intersection guaranteed by 
the Hanani-Tutte theorem may consists of two canonical curves $\sigma(a_i,v_{ij})\circ \sigma(v_{ij}, a_j)$ and $\sigma(a_c, v_{cd}) \circ \sigma(v_{cd}, a_d)$ that cross at two distinct edges (vertices of the intersection graph) $x_{ij}$ and $x_{cd}$ with the same $S$-index, meaning that $(x_{ij},x_{cd})\in E$ and have the same slope as line segments.
By rotating the plane, we may assume that $x_{ij}$ and $x_{cd}$ are horizontal line segments.
Note that both canonical curves consists of at least two vertices. 
If $x_{ij}$ (or $x_{cd}$) intersects with the predecessor/successor of $x_{cd}$ (or $x_{ij}$) on their corresponding canonical curve, we would be done (since that predecessor/successor vertex does not have identical $S$-index) by the same argument of \Cref{lm:VCbound-type}. The remainder of this proof is dedicated to showing that any such intersection given by the Hanani-Tutte theorem also guarantees an intersection between a predecessor/successor vertex.

\begin{figure}[t]
    \centering
    \includegraphics[page=1]{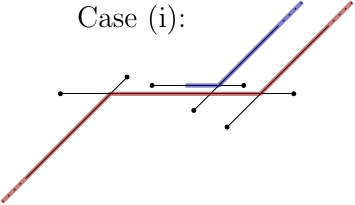}
    \includegraphics[page=2]{fig/segment_intersection2.pdf}
    \includegraphics[page=3]{fig/segment_intersection2.pdf}
    \includegraphics[page=4]{fig/segment_intersection2.pdf}
    \caption{Cases for the proof of \Cref{thm:segments}. The blue and red paths indicate the canonical curves. In all cases, the horizontal segments 
    lie on the same line, but have been drawn with a small gap in between for clarity. The predecessor/successor curves may not exist, or may also be horizontal, but we have drawn them as non-horizontal (if they exist) for clarity.}
    \label{fig:segment_intersection}
\end{figure}

See \Cref{fig:segment_intersection} for illustrations of how the canonical curves can intersect at two vertices that correspond to horizontal line segments, where a blue canonical curve can intersect with a red canonical curve at two horizontal line segments. In the figure, the predecessor/successor edges can be either horizontal or non-horizontal. We illustrate them as non-horizontal for clarity, as only the location of intersection with the predecessor/successor edges matters. 
The cases reflect how the horizontal intervals on the two canonical curves can intersect, in case (i) the red and blue horizontal segments are nested, in cases (ii) and (iv) they intersect but do not nest, and in case (iii) they are completely disjoint.
In case (i), the blue horizontal segment is nested in the red horizontal segment, and thus intersection of the blue predecessor/successor with the red horizontal segment is unavoidable. 
In case (iii) the horizontal segments of the red and blue curves are not intersecting at all, which is not a valid intersection that the Hanani-Tutte theorem can give us.
In case (ii) and (iv) the blue and red horizontal segments intersect, but are not nested. In these cases either a predecessor/successor segment intersects the horizontal segment of the other as in case (ii), or they can be separated as curves in case (iv), and are not a valid intersection that is given by the Hanani-Tutte theorem (they can be perturbed to be non-intersecting).
\end{proof}

We would like to apply \Cref{thm:segments} to the framework of \cite{ChanCGKLZ25} (more precisely \Cref{thm:diameter}) to obtain a result for \Diameter-$\Delta$ of line segments with $h$ slopes. However, there are technical issues with implementing an efficient data structure for line segments with $h$ slopes. Indeed, even for a collection of axis-aligned segments (i.e., with two slopes), we do not see how to implement the colored range searching data structure efficiently. We resolve these technical issues in \Cref{ap:h-slope-line-seg} by carefully computing neighborhood balls in the order of increasing length of sequences to reuse prior computation and ensure that the data structure problem involves exactly one type of slope at a time.
\begin{theorem} \label{thm:h-slope-line-seg}
\Diameter-$\Delta$ of line segments with at most $h$ different slopes can be solved in $n^{2-1/O(h^{\Delta+1})}$ time.  
\end{theorem}

\section{Almost-Linear Time Constant-Diameter Algorithm for 2D Unit Squares}\label{sec:unitsq:diam2}

In this section, we present an $O^*(n)$-time algorithm for testing whether the diameter of a 2D unit-square graph is at most $\Delta$, for any given constant $\Delta$.  The precise time bound is $n2^{O(\sqrt{\log n\log\log n})}$, interestingly.
All squares in this section are axis-aligned.  Our result significantly extends Bringmann et al.'s near-linear algorithm for unit-square graphs~\cite{Bringmann2022-me}, which worked only for $\Delta=2$. It also significantly improves the $\OO(n^{7/4})$ time bound of Duraj, Konieczny, and Pot\k epa's algorithm~\cite{duraj2023better} in the constant $\Delta$ case;
since Duraj et al.'s approach is based on distance VC-dimension, it does not seem possible to improve their time bound to near linear (the VC-dimension or shatter dimension for unit squares is at least 2 even for 1-neighborhoods!).

Our new algorithm employs an unusual divide-and-conquer approach:
\begin{enumerate}
\item Often in computational geometry, we divide based on (say) the median $x$- or $y$-coordinate,
but in our new algorithm, we will divide based on coordinate values modulo 1.
\item We will use a larger fan-out rather than dividing 2-way.
\item We will repeatedly map input points to new points in a much higher dimension (dependent on $\Delta$),
constructed in a nontrivial way, to enable the usage of orthogonal range searching techniques.
\end{enumerate}
Even though the diameter problem may not at first seem amenable to divide-and-conquer (not being a ``decomposable'' problem),
we show that a combination of these ideas is powerful enough to solve the problem efficiently
for any constant $\Delta$ (unlike the far more straightforward approach by Bringmann et al.\ for $\Delta=2$).

For a point $p\in\R^2$, denote by  $\cube{p}$ the unit square centered at $p$.
Let $x(p)$ and $y(p)$ denote the $x$- and $y$-coordinate of a point $p$ respectively.  
To keep the presentation simple, we assume that all coordinate values are distinct.  It is 
straightforward to modify the algorithm below to handle cases when some values may be equal, just by being more careful about $<$ vs.\ $\le$ signs and open vs.\ closed intervals.  (This is unlike our other results on axis-aligned line segments, where degeneracies can make a big difference.)
We will solve the problem in a slightly more general setting for $\Delta+1$ point sets $P_0,\ldots,P_\Delta$: testing whether all distances between $\cube{p_0}$ and $\cube{p_\Delta}$ for $(p_0,p_\Delta)\in P_0\times P_\Delta$
are exactly $\Delta$ in the $(\Delta+1)$-partite intersection graph of the unit squares centered at $P_0,\ldots,P_\Delta$,
i.e., the subgraph of the intersection graph where we include an edge between $\cube{p}$ and $\cube{q}$ only when $p\in P_j$ and $q\in P_{j+1}$ for some $j\in\{0,\ldots,\Delta-1\}$.
The original problem reduces to the case when $P_0,\ldots,P_\Delta$ are identical sets (or almost identical, if we want to ensure general position).

Build a uniform grid over $\R^2$ of unit side length.
We first identify a special case which we can solve quickly.  Specifically, suppose that the points $p_j\in P_j$ lie in a fixed grid cell for a fixed $j$, and points in $p_{j+1}\in P_{j+1}$ lie in a fixed neighboring grid cell, w.l.o.g., say, to its right; furthermore, suppose that all such points satisfy $x(p_j)\bmod 1 > \mu$ and $x(p_{j+1})\bmod 1<\mu$.  Then we already know that $|x(p_j)-x(p_{j+1})|<1$, and so it suffices to  enforce only the constraint that $|y(p_j)-y(p_{j+1})|<1$.
Consequently, we can proceed greedily: for each $p_0\in P_0$, we only need to find the lowest point in $P_j$ reachable from $p_0$, and for each $p_\Delta\in P_\Delta$, we only need to find the highest point in $P_{j+1}$ reachable from $p_\Delta$ (or vice versa, depending on the position of neighboring cell),  as explained in the lemma below:

\begin{lemma}\label{lem:unitsq1}
Given $\Delta+1$ point sets $P_0,\ldots,P_\Delta\subset \R^2$ of total size $n$, and given $j\in\{0,\ldots,\Delta-1\}$ and $\mu\in(0,1)$,
we can compute mappings $\phi: P_0\rightarrow\R$
and $\psi: P_\Delta\rightarrow\R$ in $\OO(n)$ time, satisfying the following property for every $(p_0,p_\Delta)\in P_0\times P_\Delta$: 
\begin{quote}
$(\exists (p_1,\ldots,p_{\Delta-1})\in P_1\times\cdots\times P_{\Delta-1}:$\ $\cube{p_i}$ intersects $\cube{p_{i+1}}$ for all $i$, and\\
   $p_j\in (\mu,1)\times(0,1)$, and $p_{j+1}\in ((0,\mu)\cup(1,1+\mu))\times(-1,1))$\\[.5ex]  $\Longleftrightarrow\ \ \phi(p_0)<\psi(p_\Delta)$.
\end{quote}
\end{lemma}
\begin{proof}
Define the weight of each point in $P_j\cap ((\mu,1)\times(0,1))$ to be its $y$-coordinate; the weights of all other points in $P_j$ is set to $\infty$.
For each $i=j-1,\ldots,0$, define the weight of each point $p\in P_i$ to be the min weight among all points in
$P_{i+1}$ with $L_\infty$-distance at most 1 from $p$.   These values can be computed by orthogonal range min queries.
For each $p\in P_0$, define $\phi(p)$ to be the weight of $p$.

Similarly, define the weight of each point in $P_{j+1}\cap ((1,1+\mu)\times(-1,1))$ to be its $y$-coordinate; the weights of all other points in $P_{j+1}$ is set to $-\infty$.
For each $i=j+2,\ldots,\Delta$, define the weight of each point $p\in P_i$ to be the max weight among all points in
$P_{i-1}$ with $L_\infty$-distance at most 1 from $p$.   These values can be computed by orthogonal range min queries.
For each $p\in P_\Delta$, define $\psi(p)$ to be the weight of $p$ plus 1.  

Then the property is satisfied ($p_j\in (\mu,1)\times(0,1)$ and $p_{j+1}\in ((0,\mu)\cup(1,1+\mu))\times(-1,1)$
imply $|x(p_j)-x(p_{j+1})|<1$ and $y(p_{j+1})-y(p_j)<1$, so we only need $y(p_j)-y(p_{j+1})<1$).
\end{proof}

Suppose $P_0$ lies in a fixed grid cell, say, $(0,1)^2$.
We apply \Cref{lem:unitsq1} to solve the above special case for each of the $\Delta$ choices for $j$ and each of the $O(\Delta^2)$ choices of grid cells for $P_j$ and its neighboring cells for $P_{j+1}$.
However, we cannot solve these cases separately, since the diameter problem is not decomposable.  Instead, our idea is to view each special case as imposing an additional constraint in a new dimension, so that the combined problem can be reduced to an orthogonal range searching problem~\cite{AgarwalE99,CGAA} in a higher ($O(\Delta^3)$) dimension.

\begin{lemma}\label{lem:unitsq3}
Given $\Delta+1$ point sets $P_0,\ldots,P_\Delta\subset \R^2$ of total size $n$ where $P_0\subset (0,1)^2$, and given 
$\mu\in(0,1)$,
we can compute mappings $\phi: P_0\rightarrow\R^{O(\Delta^3)}$
and $\psi: P_\Delta\rightarrow\R^{O(\Delta^3)}$ in $\OO(\Delta^3 n)$ time, satisfying the following for every $(p_0,p_\Delta)\in P_0\times P_\Delta$: 
\begin{quote}
$(\exists (p_1,\ldots,p_{\Delta-1})\in P_1\times\cdots\times P_{\Delta-1},\exists j:$\ $\cube{p_i}$ intersects $\cube{p_{i+1}}$ for all $i$, and\\
   $x(p_j)\bmod 1 > \mu$ and $x(p_{j+1})\bmod 1 < \mu$, or vice versa)\\[.5ex] $\Longleftrightarrow$\ \ $\phi(p_0)$ does not dominate $\psi(p_\Delta)$.
\end{quote}
\end{lemma}
\begin{proof}
For each $\alpha\in \Z^2$ with $L_\infty$-norm at most $\Delta$, and each $j\in\{0,\ldots,\Delta-1\}$,
we compute mappings $\phi^{(\alpha,j,1)},\ldots,\phi^{(\alpha,j,4)}: P_0\rightarrow\R$ and $\psi^{(\alpha,j,1)},\ldots,\psi^{(\alpha,j,4)}: P_\Delta\rightarrow\R$ satisfying the
following properties:
\begin{enumerate}
\item $(\exists (p_1,\ldots,p_{\Delta-1})\in P_1\times\cdots\times P_{\Delta-1}:$\ $\cube{p_i}$ intersects $\cube{p_{i+1}}$ for all $i$, and\\
   $p_j\in \alpha+((\mu,1)\times(0,1))$, and $p_{j+1}\in \alpha+(((0,\mu)\cup(1,1+\mu))\times(-1,1)))$\\[.5ex] $\Longleftrightarrow\ \ \phi^{(\alpha,j,1)}(p_0)<\psi^{(\alpha,j,1)}(p_\Delta)$.
\item $(\exists (p_1,\ldots,p_{\Delta-1})\in P_1\times\cdots\times P_{\Delta-1}:$\ $\cube{p_i}$ intersects $\cube{p_{i+1}}$ for all $i$, and\\
   $p_j\in \alpha+((\mu,1)\times(-1,0))$ and $p_{j+1}\in \alpha+(((0,\mu)\cup (1,1+\mu))\times(-1,1)))$\\[.5ex] $\Longleftrightarrow\ \ \phi^{(\alpha,j,2)}(p_0)<\psi^{(\alpha,j,2)}(p_\Delta)$.
\item $(\exists (p_1,\ldots,p_{\Delta-1})\in P_1\times\cdots\times P_{\Delta-1}:$\ $\cube{p_i}$ intersects $\cube{p_{i+1}}$ for all $i$, and\\
   $p_{j+1}\in \alpha+((\mu,1)\times(0,1))$, and $p_j\in \alpha+(((0,\mu)\cup(1,1+\mu))\times(-1,1)))$\\[.5ex] $\Longleftrightarrow\ \ \phi^{(\alpha,j,3)}(p_0)<\psi^{(\alpha,j,3)}(p_\Delta)$.
\item $(\exists (p_1,\ldots,p_{\Delta-1})\in P_1\times\cdots\times P_{\Delta-1}:$\ $\cube{p_i}$ intersects $\cube{p_{i+1}}$ for all $i$, and\\
   $p_{j+1}\in \alpha+((\mu,1)\times(-1,0))$ and $p_j\in \alpha+(((0,\mu)\cup (1,1+\mu))\times(-1,1)))$\\[.5ex] $\Longleftrightarrow\ \ \phi^{(\alpha,j,4)}(p_0)<\psi^{(\alpha,j,4)}(p_\Delta)$.
\end{enumerate}
The first two properties differ only in the location of $p_j\in \alpha+((\mu, 1)\times (0,1))$ versus $p_j\in \alpha+((\mu, 1)\times (-1,0))$, and the second two properties are identical to the first with the locations of $p_j$ and $p_{j+1}$ swapped.
Each such mapping can be computed by \Cref{lem:unitsq1} 
(shifting by $\alpha$, possibly with $y$-coordinates negated and/or $P_0,\ldots,P_\Delta$ reversed).  
Finally, we define $\phi$ and $\psi$ as the Cartesian products of these mappings $\phi^{(\alpha,j,1)},\ldots,\phi^{(\alpha,j,4)}$ 
and $\psi^{(\alpha,j,1)},\ldots,\psi^{(\alpha,j,4)}$ over all $O(\Delta^2)$ choices of $\alpha$
and $\Delta$ choices of $j$.
\end{proof}

How can we reduce the problem in the general case to the case from \Cref{lem:unitsq3}?
A natural approach is binary divide-and-conquer. 
We divide the input sets into two parts: points with $x$-coordinates less than $\mu$ modulo 1, and points with $x$-coordinates greater than $\mu$ modulo 1, where $\mu$ is the median of the $x$-coordinates mod 1.
There are 3 cases for a path $p_0,\ldots,p_\Delta$:
(i) it stays entirely in $\{(x,y): x\bmod 1 < \mu\}$;
(ii) it stays entirely in $\{(x,y): x\bmod 1 > \mu\}$;
(iii) it crosses between $\{(x,y): x\bmod 1 > \mu\}$
and $\{(x,y): x\bmod 1 < \mu\}$.
Case (i) can be handled by one recursive call with half of the points.
Case (ii) can be handled by another recursive call with the other half of the points.
If case (iii) occurs, there must be some $j$ for which $x(p_j)\bmod 1> \mu$ and $x(p_{j+1})\bmod 1<\mu$ or vice versa%
---this is \emph{precisely} the case handled by \Cref{lem:unitsq3}!

However, there is one important issue: the diameter problem is not decomposable.
For example, consider a pair $(p_0,p_\Delta)\in P_0\times P_\Delta$ with
$x(p_0)\bmod 1 <\mu$ and $x(p_\Delta)\bmod 1<\mu$.  It could be covered by case (i) or case (iii).  If we are solving the distance oracle problem (testing whether $p_0$ and $p_\Delta$ have distance $\Delta$ for a given query pair $(p_0,p_\Delta)$), we could consider both cases, recurse, and take the ``or'' of the answers, thus obtaining logarithmic query time.  However, for the diameter problem, we need to exclude pairs $(p_0,p_\Delta)$ that already pass the test from case (iii) when we make the recursive call to handle case (i).  We cannot afford to enumerate all such pairs since there are quadratically many in the worst case. Our idea is to realize that the pairs we want to exclude are encodable by orthogonal range searching/dominance constraints in $O(1)$ dimensions, according to \Cref{lem:unitsq3}.  As we recurse, we just add more and more dimensions to the encoding.

A binary divide-and-conquer has $O(\log n)$ recursion depth, and if the dimension were to increase to $O(\log n)$, we would no longer be able to solve the orthogonal range searching subproblems in near-linear time.
A final idea is to use a larger fan-out $b$ for the divide-and-conquer (for example, the recursion depth is reducible to $O(1)$ by setting $b=n^\eps$, though we will suggest a slightly better choice of $b$ below).
The details are made precise below:

\begin{lemma}\label{lem:unitsq4}
Given $\Delta+1$ point sets $P_0,\ldots,P_\Delta\subset \R^2$ of total size $n$, 
where $P_0\subset (0,1)^2$,
and given parameter $D$ and mappings $f: P_0\rightarrow \R^D$ and $g: P_\Delta\rightarrow \R^D$,
we can decide whether for all $(p_0,p_\Delta)\in P_0\times P_\Delta$ with $f(p_0)$ dominating $g(p_\Delta)$,
there exists $(p_1,\ldots,p_{\Delta-1})\in P_1\times\cdots\times P_{\Delta-1}$, such that
$\cube{p_i}$ intersects $\cube{p_{i+1}}$ for all $i$, 
in $n2^{O(\sqrt{\Delta^3\log n\log\log n})}\log^Dn$ time.
\end{lemma}
\begin{proof}
We use $b$-way divide-and-conquer for a parameter $b$ to be set later:
\begin{enumerate}
\item Compute quantiles $0=\mu_0 <\mu_1< \cdots < \mu_{b-1} < \mu_b=1$ such that each interval $(\mu_k,\mu_{k+1})$ contains roughly
$n/b$ elements of $\{x(p)\bmod 1: p\in P_0\cup\cdots\cup P_\Delta\}$.  For each $k\in\{0,\ldots,b-1\}$ and $i\in\{0,\ldots,\Delta\}$, let $P_i^{(k)}=\{p\in P_i: x(p)\bmod 1\in (\mu_k,\mu_{k+1})\}$.
\item For each $k\in \{0,\ldots,b\}$, compute mappings $\phi^{(k)}: P_0\rightarrow\R^{O(\Delta^3)}$
and $\psi^{(k)}: P_\Delta\rightarrow\R^{O(\Delta^3)}$, satisfying the following property:
\begin{quote}
$(\exists (p_1,\ldots,p_{\Delta-1})\in P_1\times\cdots\times P_{\Delta-1},\exists j:$\ $\cube{p_i}$ intersects $\cube{p_{i+1}}$ for all $i$, and\\
   $x(p_j)\bmod 1 > \mu_k$ and $x(p_{j+1})\bmod 1 < \mu_k$, or vice versa)\\[.5ex] $\Longleftrightarrow$\ \ $\phi^{(k)}(p_0)$ does not dominate $\psi^{(k)}(p_\Delta)$.
\end{quote}
\item
For each $k\in\{0,\ldots,b-1\}$, verify that there is no pair $(p_0,p_\Delta)\in P_0^{(k)}\times (P_\Delta\setminus P_\Delta^{(k)})$
such that $(f(p_0),\phi^{(k)}(p_0),\phi^{(k+1)}(p_0))$ dominates 
$(g(p_0),\psi^{(k)}(p_\Delta),\psi^{(k+1)}(p_\Delta))$.
This can be done by $(D+O(\Delta^3))$-dimensional orthogonal range searching in $bn(\log n)^{D+O(\Delta^3)}$ time.

\item
For each $k\in\{0,\ldots,b-1\}$, recursively solve the problem for $P_0^{(k)},\ldots,P_\Delta^{(k)}$ with $f$ and $g$ replaced by
$(f,\phi^{(k)},\phi^{(k+1)})$ and $(g,\psi^{(k)},\psi^{(k+1)})$.
\end{enumerate}

To see correctness, observe that for any $(p_0,\ldots,p_\Delta)\in P_0\times\cdots\times P_\Delta$ and any $k$, one of the following is true:
(i)~$\exists j: x(p_j)\bmod 1 < \mu_k$ and $x(p_{j+1})\bmod 1 > \mu_k$, or vice versa, or 
(ii)~$\exists j: x(p_j)\bmod 1 < \mu_{k+1}$ and $x(p_{j+1})\bmod 1 > \mu_{k+1}$, or vice versa, or
(iii)~$(p_0,\ldots,p_\Delta)$ is in $P_0^{(k)}\times\cdots\times P_\Delta^{(k)}$, or
(iv)~$(p_0,\ldots,p_\Delta)$ is in $(P_0\setminus P_0^{(k)})\times\cdots\times (P_\Delta\setminus P_\Delta^{(k)})$.  (In other words,
if the sequence of $x$-coordinates in $(p_0,\ldots,p_\Delta)$ modulo 1 does not
``cross'' $\mu_k$ nor $\mu_{k+1}$, then the sequence must stay completely inside the interval $(\mu_k,\mu_{k+1})$ or completely outside.)
For $(p_0,p_\Delta)\in P_0^{(k)}\times (P_\Delta\setminus P_\Delta^{(k)})$, it is (i) or (ii) (handled by step~3).
For $(p_0,p_\Delta)\in P_0^{(k)}\times P_\Delta^{(k)}$, it is (i), (ii), or (iii) (handled by step~4).

The running time satisfies the recurrence $T(n,D)\le b\,T(n/b,D+O(\Delta^3))+bn(\log n)^{D+O(\Delta^3)}$.
There are $O(\log_b n)$ levels of recursion, and the dimension parameter $D$ increases by $O(\Delta^3)$ at each level.
The recurrence solves to $T(n,D)\le bn(\log n)^{D+O(\Delta^3\log_b n)}$.  Setting $b=2^{\sqrt{\Delta^3\log n\log\log n}}$ yields
$T(n,D)\le n2^{O(\sqrt{\Delta^3\log n\log\log n})}\log^Dn$.
\end{proof}


\begin{theorem}\label{thm:unitsq}
Given $\Delta+1$ point sets $P_0,\ldots,P_\Delta\subset \R^2$ of total size $n$, 
we can decide whether for all $(p_0,p_\Delta)\in P_0\times P_\Delta$,
there exists $(p_1,\ldots,p_{\Delta-1})\in P_1\times\cdots\times P_{\Delta-1}$, such that
$\cube{p_i}$ intersects $\cube{p_{i+1}}$ for all $i$, 
in time $n2^{O(\sqrt{\Delta^3\log n\log\log n})}$, which is $O^*(n)$ 
for any constant $\Delta$ (or, in fact, for any $\Delta=o(\log^{1/3}n/\log^{1/3}\log n)$).
\end{theorem}
\begin{proof}
Build a uniform grid of side length 1.
For each nonempty grid cell $\alpha+(0,1)^2$ (with $\alpha\in\Z^2$),
we solve the problem for
$P_0\cap (\alpha+(0,1)^2)$, $P_1\cap (\alpha+(-1,2)^2)$, \ldots,
$P_\Delta\cap (\alpha+(-\Delta,\Delta+1)^2)$
by \Cref{lem:unitsq4} (with trivial mappings $f$ and $g$, i.e., $D=0$).
Each point participates in $O(\Delta^2)$ subproblems.
\end{proof}

\renewcommand{\L}{\mathrm{left}}
\renewcommand{\R}{\mathrm{right}}

\section{Diameter-2 Lower Bound for 2D Triangles}\label{sec:triangle}

In this and the next three sections, we turn to proving conditional lower bounds.
We begin with a proof of a near-quadratic lower bound for \Diameter-2 for 2D triangles, assuming the 3-uniform 6-hyperclique hypothesis.

Consider convex chains $C_\L,C_\R$ on six points each (i.e., including endpoints), positioned so that for any point pair $p_\L\in C_\L, p_\R\in C_\R$, the segment $p_\L p_\R$ is disjoint from the interiors of the convex hulls of $C_\L$ and $C_\R$. See \Cref{fig:triangle_basic} for an example of such a configuration.

Let $v^1,\dots,v^6$ denote the vertices along $C_\L$, and $s_{AB},s_{BC},s_{CA}$ denote the segments $v^1v^2,v^3v^4$, and $v^5v^6$, respectively.  We define the segments $s_{DE},s_{EF},s_{FD}$ analogously along $C_\R$.

\begin{figure}[tb]
    \centering
    \includegraphics{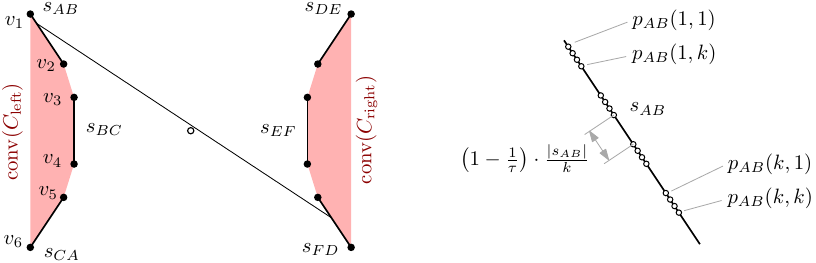}
    \caption{Left: Two convex chains $C_\L$ and $C_\R$ placed so that any segment connecting them remains disjoint from the interiors of their convex hulls (the red shaded regions). Right: Defining a group of $k\times k$ points along $s_{AB}$, such that $p_{AB}(a,b)$ are grouped by the value of $a$, and groups have gaps of size $\simeq|s_{AB}|/k$.}
    \label{fig:triangle_basic}
\end{figure}

Let $G=(V_A\cup V_B\cup V_C\cup V_D\cup V_E\cup V_F,E)$ be a 6-partite 3-uniform hypergraph with $k:=n^{1/3}$ vertices in each of the six parts. Our idea is to encode some edges and non-edges of this graph via certain triangles. First, we associate specific points on the previously defined segments with certain pairs $(a,b)\in [k]^2$. We will then encode a triplet $(v_a,v_b,v_c)\in V_A\times V_B \times V_C$ via the pairs $(a,b)$, $(b,c)$ and $(c,a)$.

Consider now the segment $s_{AB}$. For each $a,b\in [k]^2$ let $p_{AB}(a,b)$ be the point of $s_{AB}$ at distance
\[|s_{AB}|\cdot\left(\frac{a-1}{k}+\frac{b}{\tau k^2}\right)\]
from $v_1$, where $|s_{AB}|$ denotes the length of $s_{AB}$, and $\tau\geq 2$ is a constant we can choose freely; for now, we can set $\tau=2$. Similarly, for each $(b,c)\in [k]^2$ and $(c,a)\in [k]^2$ we set $p_{BC}(b,c)$ and $p_{CA}(c,a)$ to be the point of $s_{BC}$ at distance $|s_{BC}|\cdot(\frac{b-1}{k}+\frac{c}{\tau k^2})$ from $v_3$ and the point of $s_{CA}$ at distance $|s_{CA}|\cdot(\frac{c-1}{k}+\frac{a}{\tau k^2})$ from $v_5$, respectively. Define the points $p_{DE}(d,e)$, $p_{EF}(e,f)$ and $p_{FD}(f,d)$ along the segments $s_{DE},s_{EF},s_{FD}$ analogously.

In each part $V_X$ of $G$, $X\in \{A, B, C, D, E, F\}$ we index the vertices from $1$ to $k$, i.e., $V_X=\{v^X_1,\dots,v^X_k\}$. We will drop the superscripts of the vertices when they can be inferred from the context.
Let $\{X,Y,Z\}$ be three distinct symbols among $\{A,B,C,D,E,F\}$. The triplet $\{X,Y,Z\}$ is considered a \emph{crossing triplet} if $\{X,Y,Z\}$ has a non-empty intersection with both $\{A,B,C\}$ and $\{D,E,F\}$.

Our reduction realizes an intersection graph similar to the intersection graph for a lower bound of~\cite{Bringmann2022-me} designed for $12$-dimensional hypercubes.

Let $G'$ be the intersection graph of the following triangles (see also \Cref{fig:triangle_thin}):
\begin{itemize}
    \item For each edge $(v_a,v_b,v_c)\in E(G)\cap (V_A\times V_B\times V_C)$ add the triangle
    \[\Delta_{ABC}(a,b,c)\text{ with vertices } p_{AB}(a,b),p_{BC}(b,c), p_{C,A}(c,a).\]
    \item For each edge $(v_d,v_e,v_f)\in E(G)\cap (V_D\times V_E\times V_F)$ add the triangle
    \[\Delta_{DEF}(d,e,f)\text{ with vertices } p_{DE}(d,e), p_{EF}(e,f),p_{FD}(f,d).\]
    \item For each crossing triplet $\{X,Y,Z\}$ with $XY\in \{AB,BC,CA\}$ and $ZW\in \{DE,EF,FD\}$ and each \emph{non-edge} $(v_x,v_y,v_z)\in (X\times Y\times Z) \setminus E(G)$ add a triangle
    \[\Delta_{XYZ}(x,y,z)\text{ with vertices }p_{XY}(x,y)p_{ZW}(z,1)p_{ZW}(z,k).\]
    \item For each crossing triplet $\{X,Y,Z\}$ with $XY\in \{DE,EF,FD\}$ and $ZW\in \{AB,BC,CA\}$  and each \emph{non-edge} $(v_x,v_y,v_z)\in (X\times Y\times Z) \setminus E(G)$ add a triangle
    \[\Delta_{ZXY}(z,x,y)\text{ with vertices }p_{ZW}(z,1)p_{ZW}(z,k)p_{XY}(x,y).\]
    \item Add a triangle $\Delta_{\L}$ that covers all triangles $\Delta_{ABC}$, intersects all triangles corresponding to crossing triplets, and is in the open half-plane $x<0$. Let $\Delta_\R$ be the mirror image of $\Delta_\L$.
\end{itemize}

\begin{figure}[tb]
    \centering
    \includegraphics{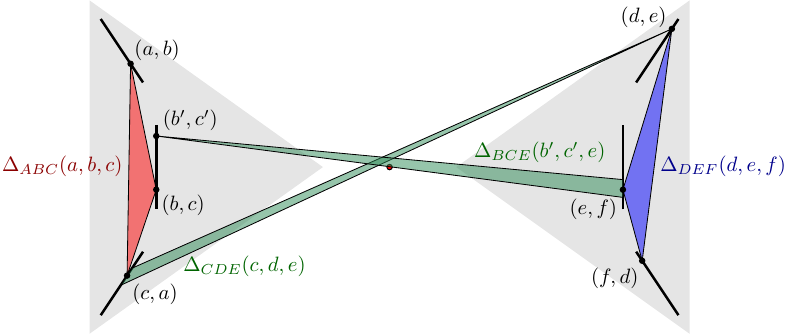}
    \caption{The construction of the intersection graph $G'$. Triangles on the left and right (red and blue) correspond to edges of $G$, crossing triangles (green) correspond to non-edges of $G$. The two gray triangles are the dummy triangles $\Delta_\L$ and $\Delta_\R$.}
    \label{fig:triangle_thin}
\end{figure}

\begin{observation}\label{obs:triangle_lower} The vertices of the constructed graph $G'$ satisfy the following properties:
    \begin{enumerate}
        \item $\Delta_{ABC}(a,b,c)\cap \Delta_{ABD}(a',b',d) \neq \emptyset$ if and only if $a=a'$ and $b=b'$.
        \item $\Delta_{ABC}(a,b,c)\cap \Delta_{ADE}(a',d,e) \neq \emptyset$ if and only if $a=a'$.
        \item $\Delta_{ABC}(a,b,c)\cap \Delta_{DEF}(d,e,f)=\emptyset$.
        \item $\Delta_{ABD}(a,b,d)\cap \Delta_{\L}\neq \emptyset$ and $\Delta_{ABD}(a,b,d)\cap \Delta_{\R}\neq \emptyset$.
        \item $\Delta_{ABC}(a,b,c)\cap \Delta_{\R} = \emptyset$ and $\Delta_{ABC}(a,b,c)\cap \Delta_{\L} \neq \emptyset$.
    \end{enumerate}
    The analogous statements hold when exchanging $ABC$ with $DEF$ or exchanging the crossing triplets with other crossing triplets.
\end{observation}

We can now prove the desired lower bound.

\begin{theorem}\label{thm:lb_thintriangle}
    Assuming the 3-uniform 6-hyperclique hypothesis,  there is no $O(n^{2-\eps})$ time algorithm for deciding if the intersection graph of a given set of $n$ triangles in the Euclidean plane has diameter at most 2. 
\end{theorem}

\begin{proof}
Bringmann et al.~\cite{Bringmann2022-me} show that an intersection graph that satisfies Conditions 1-3 as described in Observation~\ref{obs:triangle_lower} has a pair of vertices $\Delta_{ABC}(a,b,c),\Delta_{DEF}(d,e,f)$ with hop distance greater than $2$ if and only if $(v_a,v_b,v_c,v_d,v_e,v_f)$ is a hyperclique of $G$.
We recall the argument briefly. Consider a pair of triangles $\Delta_1=\Delta_{ABC}(a,b,c)$ and $\Delta_2=\Delta_{DEF}(d,e,f)$. They both correspond to edges, so we have that the tuple $(v_a,v_b,v_c,v_d,v_e,v_f)$ forms a hyperclique if and only if all crossing triplets among $(v_a,v_b,v_c)$ and $(v_d,v_e,v_f)$ are present in $G$. Since we only represent non-crossing triplets in the graph, the tuple forms a hyperclique if and only if there is no crossing triangle $\Delta_{XYZ}$ connecting $\Delta_1$ and $\Delta_2$, i.e., they have no shared neighbor, or equivalently, they have hop distance at least $3$.

The reduction is almost complete. However, we need to ensure that no other vertex pairs can form a pair whose distance is at least $3$ in $G'$. 
Let $D_\L$ denote the set of triangles $\Delta_{ABC}(a,b,c)$, and let $D_\R$ denote the set of triangles $\Delta_{DEF}(d,e,f)$. One can verify that the triangles $\Delta_\L,\Delta_\R$ ensure that if a given pair $\Delta,\Delta'$ of vertices of $G'$ satisfies (a) $\Delta,\Delta'\in D_\L$, or (b) $\Delta,\Delta'\in D_\R$, or (c) at least one of $\Delta,\Delta'$ is not in $D_\L\cup D_\R$, then the hop distance of $\Delta$ and $\Delta'$ is at most $2$. Thus the only setting where distance $3$ is possible is the one discussed above.
\end{proof}

We can strengthen \Cref{thm:lb_thintriangle} to \emph{fat} triangles; in fact the angles of every triangle can be forced to be close to a right-angled isoceles triangle, henceforth called \EMPH{half-squares}. For a fixed $\eps>0$ we say that a triangle is an \EMPH{$\eps$-half-square} if its longest side length is in the interval $[(\sqrt{2}-\eps)s,(\sqrt{2}+\eps)s]$ and its two shorter side lengths fall in the interval $[(1-\eps)s,(1+\eps)s]$ for some $s>0$.

We claim that the above intersection graph can be realized with $\eps$-half-squares for any fixed $\eps>0$, as follows. First, we choose both convex chains in such a way that the segments $s_{XY}$ for $\{XY\}\in \{AB,BC,CE\}$ have length $\eps/100$ and there is a triangle with base $v_2v_5$ of length $\sqrt{2}+\eps/2$ whose third vertex is the midpoint of $s_{BC}=v_3v_4$. We place the chain as depicted in \Cref{fig:triangle_fat}(i) so that it is symmetric on the $x$ axis, and the line through $s_{AB}$ and $s_{CA}$ have slope $-1-\eps/10$ and $1+\eps/10$, respectively, and they intersect the $x$-axis at some point $(-T,0)$, where $T=T(\eps)$ will be determined later. The segments $s_{DE},s_{EF},s_{FD}$ are obtained by reflecting $s_{AB},s_{BC},s_{CA}$ on the origin. Observe that the triangles $\Delta_{ABC}(a,b,c)$ and $\Delta_{DEF}(d,e,f)$ are $\eps$-half-squares. For all other triangles, we modify them so that they are half-squares, as follows:
\begin{itemize}
    \item For each crossing triplet $\{X,Y,Z\}$ with $XY\in \{AB,BC,CA\}$ and $ZW\in \{DE,EF,FD\}$ and each \emph{non-edge} $(v_x,v_y,v_z)\in X\times Y\times Z \setminus E(G)$ add the unique half-square
    $\Delta_{XYZ}(x,y,z)$ with base $uu'$ and third vertex $\hat u$ as follows:
    \begin{itemize}
        \item We set $u=p_{XY}(x,y)$,
        \item $u'$ is chosen so that the intersection of $\Delta_{XYZ}(x,y,z)$ with $s_{ZW}$ is the segment $p_{XW}(z,1)p_{ZW}(z,k)$, and 
        \item $\hat u$ lies in the half-plane $y>0$.
    \end{itemize}
    \item For each crossing triplet $\{X,Y,Z\}$ with $XY\in \{DE,EF,FD\}$ and $ZW\in \{AB,BC,CA\}$  and each \emph{non-edge} $(v_x,v_y,v_z)\in X\times Y\times Z \setminus E(G)$ we modify
    $\Delta_{ZXY}(z,x,y)$ analogously.
    \item The triangles $\Delta_\L,\Delta_\R$ are now half-squares with the properties expressed in Observation~\ref{obs:triangle_lower}.
\end{itemize}

\begin{figure}[t]
    \centering
    \includegraphics[width=\linewidth]{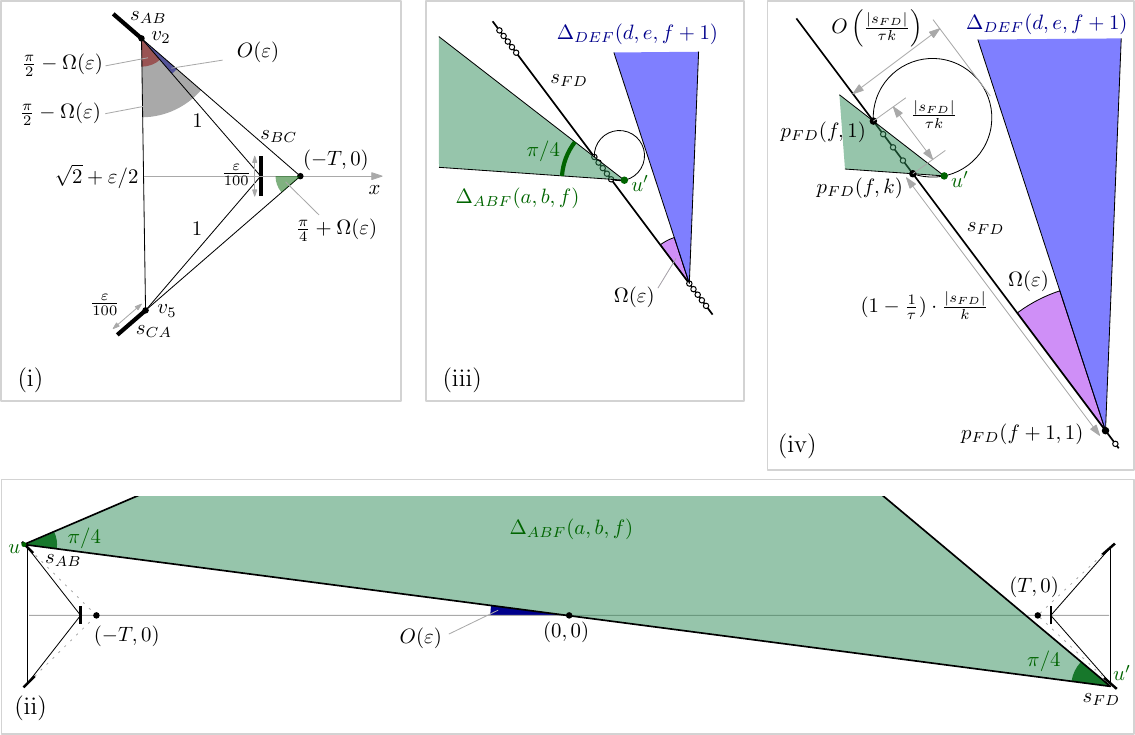}
    \caption{(i) Placement of the segments $s_{AB},s_{BC},s_{CA}$. (ii) By choosing $T$ large enough we ensure that the base of the half-square $\Delta_{ABF}(a,b,f)$ has angle $O(\eps)$ with the $x$-axis. (iii) The intersection of $\Delta_{ABF}(a,b,f)$ with the segment $s_{FD}$. We need to ensure that $\Delta_{ABF}(a,b,f)$ remains disjoint from $\Delta{DEF}(d,e,f')$ for all $f'\neq f$. (iv) The vertex $u'$ of $\Delta_{ABF}(a,b,f)$ is on the arc of inscribed angle $\pi/4$ with respect to the segment $p_{FD}(f,1)p_{FD}(f,k)$. By choosing $\tau=\Omega(1/\eps)$ large enough, we ensure that the arc is disjont from $\Delta_{DEF}(d,e,f+1)$ for all $d,e\in [k]$.}
    \label{fig:triangle_fat}
\end{figure}

We need the following geometric properties to establish Observation~\ref{obs:triangle_lower} for this modified construction.
\begin{lemma}\label{lem:fattriangle_lb}
    For every $\eps>0$ there exist $\tau$ and $T$ such that:
    \begin{enumerate}        
        \item[(1)] The intersection of $\Delta_{XYZ}(x,y,z)$ and $\conv(C_\L)\cup \conv(C_\R)$ is\\ the triangle $u'p_{ZW}(z,1)p_{ZW}(z,k)$ and the point $u$.
        \item[(2)] The triangle $u'p_{ZW}(z,1)p_{ZW}(z,k)$ is disjoint from all triangles corresponding to edges (i.e., the triangles $\Delta_{ABC}$ and $\Delta_{DEF}$) that are not incident to $z\in V_Z$.
    \end{enumerate}
\end{lemma}

\begin{proof}
    First, notice that the placement of the segments (\Cref{fig:triangle_fat}(i)) ensures that $s_{AB}$, $s_{CA}$, $s_{DE}$,and $s_{FD}$ have angle $\pi/4+c_0\eps$ with the $x$-axis for some $c_0>0$. Thus, if we set $T=c_1/\eps$ for a large enough $c_1$, then any segment connecting the left and right convex chain has angle at most $\frac{c_0}{2}\eps$ with the $x$-axis, and thus the segment directions of $s_{ZW}$ ($ZW\in \{DE,EF,FD\}$) do not fall into the angle cone of $u'$ for any crossing triangle $\Delta_{XYZ}(x,y,z)$. The symmetric claim holds for crossing triangles that intersect one of $s_{AB},s_{BC},s_{CD}$ in an interval of positive length. This concludes the proof of (1).
    
    Due to the placement of the segments (\Cref{fig:triangle_fat}(i)) we have that any side of the triangle $\Delta_{DEF}(d,e,f)$ and any of the segments $s_{DE},s_{EF},s_{FD}$ have angle at least $c_2\eps$ for some $\eps>0$ (\Cref{fig:triangle_fat}(ii)). Observe that for the vertex $u'$ of the crossing triangle $\Delta_{DEF}(d,e,f)$ the angle subtended by $p_{FD}(f,1)p_{FD}(f,k)$ from $u'$ is $\pi/4$, thus $u'$ is on a circular arc that has diameter $O(|p_{FD}(f,1)p_{FD}(f,k)|)=O(|s_{FD}|/(\tau k))$ (see \Cref{fig:triangle_fat}(iii)). On the other hand, the distance between the segment $p_{FD}(f,1)p_{FD}(f,k)|$ and any other point $p_{FD}(f',d)$ for $f'\neq f$ is at least $|s_{FD}|/(2 k)$ when $\tau\geq 2$. Consequently, the distance of any point on the segment $p_{FD}(f,1)p_{FD}(f,k)|$ and any triangle $\Delta_{DEF}(d,e,f')$ where $f'\neq f$ is at least $\Omega(\eps|s_{FD}|/ k)$ (see \Cref{fig:triangle_fat}(iv)). Thus, if we set $\tau=c_3/\eps$ for some large enough $\tau$, then $u'p_{FD}(f,1)p_{FD}(f,k)$ is disjoint from $\Delta_{DEF}(d,e,f')$ for all $f'\neq f$. This concludes the proof of (2).
\end{proof}

With \Cref{lem:fattriangle_lb} at hand, it is straightforward to verify each property of Observation~\ref{obs:triangle_lower}. We obtain the following lower bound.

\begin{theorem}\label{thm:lb_fattriangle}
    Assuming the 3-uniform 6-hyperclique hypothesis,  there is no $O(n^{2-\eps})$ time algorithm for deciding if the intersection graph of a given set of $n$ fat triangles (or even $\delta$-half-squares for any given $\delta>0$) in the Euclidean plane has diameter at most 2. 
\end{theorem}

\section{Diameter-2 Lower Bound for Segments}\label{sec:lb:seg:diam2}

In this section, we prove a near-quadratic lower bound for \Diameter-2 for line segments for combinatorial algorithms, assuming the combinatorial 4-clique hypothesis.
Our construction uses a geometric configuration that is similar to (but a bit simpler than) our lower bound construction for triangles in \Cref{sec:triangle}. We will also rely on a very similar encoding of number pairs $[k]^2$.

Let $C_\L$ be the convex chain given by the vertices $v_1=(-3,3), v_2=(-2,1), v_3=(-2,-1), v_4=(-3,-3)$. Let $s_{AB},s_{BA}$ denote the segments $v_1v_2,v_3v_4$, respectively. We define $C_\R$ to be the reflection of $C_\L$ on the origin, and we define the segments $s_{CD},s_{DC}$ analogously along $C_\R$. See \Cref{fig:segment_lb_basic}.

Let $G=(V_A\cup V_B\cup V_C\cup V_D)$ be a 4-partite graph with $k:=n^{1/2}$ vertices in each of the four parts. Our idea is to encode some edges and non-edges of this graph via certain segments. First, we associate specific points on the previously defined segments $s_{AB},s_{BA},s_{CD},s_{DC}$ with certain pairs $(x,y)\in [k]^2$. We will then represent certain edges and non-edges of $G$ with the help of these points.

\begin{figure}
    \centering
    \includegraphics{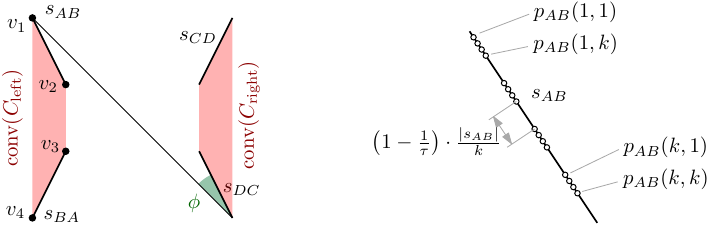}
    \caption{Left: the segments $s_{AB},s{BA},s_{CD},s_{DC}$. Any segment connecting a left and right segment among them has angle at least $\phi$ with them. Right: placing the points $p_{AB}(a,b)$ along $s_{AB}$.}
    \label{fig:segment_lb_basic}
\end{figure}

Consider first the segment $s_{AB}$. For each $(a,b)\in [k]^2$ let $p_{AB}(a,b)$ be the point of $s_{AB}$ at distance
\[|s_{AB}|\cdot\left(\frac{a-1}{k}+\frac{b}{\tau k^2}\right)\]
from $v_1$, where $|s_{AB}|$ denotes the length of $s_{AB}$, and $\tau\geq 2$ is a constant we can choose freely; for now, we can set $\tau=2$. Similarly, for each $(b,a)\in [k]^2$ we set $p_{BA}(b,a)$ to be the point of $s_{BA}$ at distance $|s_{BA}|\cdot(\frac{b-1}{k}+\frac{a}{\tau k^2})$ from $v_3$ . Define the points $p_{CD}(c,d)$, $p_{DC}(d,c)$ along the segments $s_{CD},s_{DC}$ analogously.

In each part $V_X$ of $G$, $X\in \{A, B, C, D\}$, we index the vertices from $1$ to $k$, i.e., $V_X=\{v^X_1,\dots,v^X_k\}$. We will drop the superscripts of the vertices when they can be inferred from the context.
Let $\{X,Y\}$ be distinct symbols among $\{A,B,C,D\}$. The pair $\{X,Y\}$ is considered a \emph{crossing pair} if $\{X,Y\}$ has a non-empty intersection with both $\{A,B\}$ and $\{C,D\}$. 

Let $G'$ be the intersection graph of the following segments (see \Cref{fig:segment_lb_adv}):
\begin{itemize}
    \item For each edge $(v_a,v_b)\in E(G)\cap (V_A\times V_B$) add the segment
    \[t_{AB}(a,b)\text{ with endpoints } p_{AB}(a,b),p_{BA}(b,a).\]
    We call these segments \EMPH{left segments}.
    \item For each edge $(v_c,v_d)\in E(G)\cap (V_C\times V_D)$ add the triangle
    \[t_{CD}(c,d)\text{ with endpoints } p_{CD}(c,d), p_{DC}(d,c).\]
    We call these segments \EMPH{right segments}.
    \item For each crossing pair $\{X,Y\}$ with $XX'\in \{AB,BA\}$ and $YY'\in \{CD,DC\}$ and each \emph{non-edge} $(v_x,v_y)\in (V_X\times V_Y) \setminus E(G)$ add a segment $h_{XY}(x,y)$ as defined in \Cref{lem:cross-segment} below. We call these segments \EMPH{crossing segments}.
    \item Add the four sides of the square with vertices $(-3,0),(0,-3),(3,0),(0,3)$ as well as its vertical diagonal from $(0,3)$ to $(0,-3)$. These five segments are called \emph{dummy} segments.
\end{itemize}

The crossing segments are guaranteed due to the following lemma.

\begin{lemma}\label{lem:cross-segment}
    For each crossing pair $\{X,Y\}$ with $XX'\in \{AB,BA\}$ and $YY'\in \{CD,DC\}$ and each $(x,y)\in [k]^2$ there exists a segment $h_{XY}(x,y)$ that intersects the segment $t_{XX'}(\tilde x,x')$ (resp., $t_{YY'}(\tilde y,y')$) if and only if $x=\tilde x$ (resp., $y=\tilde y$).
\end{lemma}
\begin{proof}
    We prove the lemma for $XX'=AB$ and $YY'=DC$ and the segment $h_{AD}(a,d)$; the other variants can be proven analogously. See \Cref{fig:segment_lb_adv}.
    Let $\ell_{ad}$ denote the line $p_{AB}(a,k)p_{DC}(d,1)$. Let $u$ be the intersection of $\ell_ab$ and the line $p_{AB}(a,1)v_4$. Similarly, let $u'$ denote the intersection of $\ell_ab$ and the line $p_{DC}(d,1)v'_2$, where $v'_2=(2,1)$ is the lower endpoint of $s_{CD}$. Then we define  $h_{AD}(a,d)$ to be the segment $uu'$, as depicted in \Cref{fig:segment_lb_adv}(iii).

    Let $Q_a$ denote the convex quadrilateral $p_{AB}(a,1)p_{AB}(a,k)v_3v_4$. Every segment $t_{AB}(a,b)$ for all $b\in [k]$ connects the side $p_{AB}(a,1)p_{AB}(a,k)$ of $Q_a$ to its opposite side $v_3v_4$. Let $T_a$ denote this set of segments. Consequently, any segment from $p_{AB}(a,k)$ to any point on the side $p_AB(a,1)v_4$ intersects all segments in $T_a$. In particular, it follows that  $h_{AD}(a,d)$ intersects all segments $t_{AB}(a,b)$ for $b\in [k]$. It remains to show that $h_{AD}(a,d)$ is disjoint from all segments $t_{AB}(a',b)$ where $a'\neq a$ and $a',b\in [k]$.

    Recall that any line intersecting both $s_{AB}$ and $s_{DC}$ has an angle at least $\phi>0$ with $s_{AB}$ (\Cref{fig:segment_lb_basic} left). On the other hand, any segment connecting $s_{AB}$ and $s_{BA}$ has an angle at least $\psi>0$ with $s_{AB}$ (\Cref{fig:segment_lb_adv}(ii)). Thus, two such lines have an angle of at least $\phi+\psi>0$, and in particular, $\angle \big(p_{AB}(a,k)\,u\,p_{AB}(a,1)\big) \geq \phi+\psi$, thus $u$ is within the circular arc $\gamma_a$ from which the segment $p_{AB}(a,k)p_{AB}(a,1)$ is subtended by an angle of size $\phi+\psi$. Consequently, the distance of $u$ from $s_{AB}$ is $O(|p_{AB}(a,k)p_{AB}(a,1)|) = O(|s_{AB}|/(k\tau))$. On the other hand, $\gamma_a$ has distance at least $(1-\frac1\tau)\cdot |s_{AB}|/k-O(|s_{AB}|/(k\tau))$ from any point $p_{AB}(a',b)$ when $a'\neq a$; in particular, for $\tau$ large enough, this distance is at least $|s_{AB}|/(2k)$. Thus, any line segment $t_{AB}(a',b)$ has angle at least $\psi$ with $s_{AB}$ and intersects $s_{AB}$ at a distance at least  $|s_{AB}|/(2k)$, so if $\tau$ is a large enough constant, then we have that $t_{AB}(a',b)$ is disjoint from the circle $\gamma_a$. This implies that $h_{AD}(a,d)$ is disjoint from $t_{AB}(a',b)$, concluding the proof.
\end{proof}

\begin{figure}
    \centering
    \includegraphics[width=\linewidth]{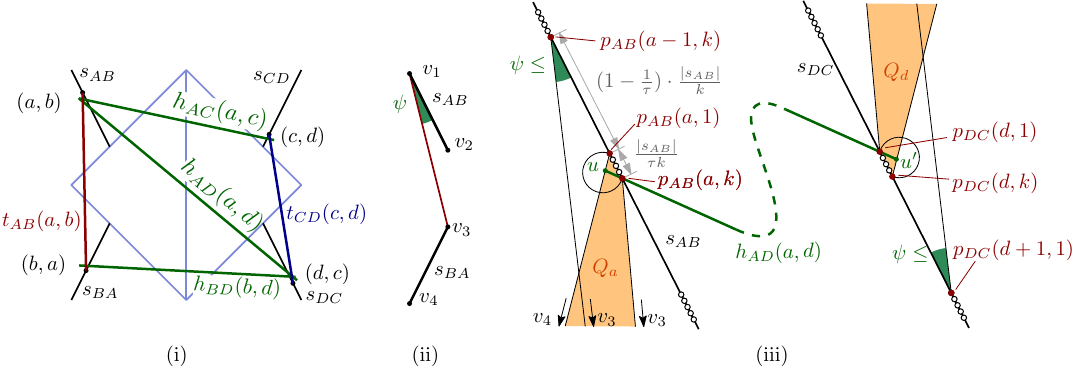}
    \caption{(i) The construction with segments, with dummy segments in light blue. The segment $h_{AD}(a,d)$  (green) corresponding to a crossing pair is lengthened beyond $p_{AB}(a,k)$ to cross all segments $t_{AB}(a,b)$ for all $b\in [k]$. (ii) Any segment $t_{AB}(a,b)$ has angle at least $\psi=\angle v_3v_1v_2$ with both $s_{AB}$ and $s_{BA}$. (iii) The two ends of the segment $h_{AD}(a,d)$ (the middle of this picture is distorted). Segment $h_{AD}(a,d)$ crosses the quadrilateral $Q_a$ (orange, left) to ensure that it intersects all segments $t_{AB}(a,b)$ for all $b\in [k]$. The angle of $h_{AD}(a,d)$ with the segment $p_{AB}(a,1)v_4$ is bounded by a constant, thus its endpoint stays inside the circle. By setting $\tau$ large enough, we ensure that $h_{AD}(a,d)$ does not intersect any segment $t_{AB}(a',b')$ where $a'\neq a$.}
    \label{fig:segment_lb_adv}
\end{figure}

We are now ready to prove our lower bound for segment intersection graphs.

\begin{theorem}\label{thm:segment_lower} Assuming the combinatorial 4-clique hypothesis,  there is no $O(n^{2-\eps})$ time combinatorial algorithm for deciding if the intersection graph of a given set of $n$ segments in the plane has diameter at most 2. 
\end{theorem} 

\begin{proof}
    Consider a pair of edges $(a,b)\in V_A\times V_B$ and $(c,d)\in V_C\times V_D$. We claim that $t_{AB}(a, b)$ and $t_{CD}(c,d)$ have hop distance at least $3$ in the constructed intersection graph if and only if $(a,b,c,d)$ is a $4$-clique in $G$. Indeed, they have hop distance at least $3$ if and only if there is no cross-edge $h_{XY}$ connecting them. Since the segments $h_{XY}$ are defined for non-edges, this is equivalent to saying that all of the edges $ac,ad,bc,bd$ are in $G$, thus $(a,b,c,d)$ is a $4$-clique.

    It remains to show that no other pair of segments can have hop distance at least $3$ in the constructed intersection graph. Indeed, one can verify that unless $s$ is a left segment and $s'$ is a right segment (or vice versa), there is always a dummy segment that is a shared neighbor of $s$ and $s'$.
    The construction takes linear time, thus the theorem follows.
\end{proof}

\section{Diameter-1 (Clique) Lower Bound for Strings}\label{sec:diameter1-string}

\begin{theorem}\label{thm:string_lower}
Assuming the OV hypothesis,  there is no $O(n^{2-\eps})$ time algorithm for deciding if the intersection graph of a given set of $n$ polygonal chains of complexity $O(\log n)$ in $\reals^2$ has diameter at most 1, or equivalently, for deciding if the induced intersection graph is a clique. 
\end{theorem}

\begin{proof}
    Let $A,B\subset \{0,1\}^d$ be a set of $n$ vectors each of dimension $d=O(\log n)$. We create the following polygonal chains:
    \begin{itemize}
        \item For each vector $u=(u_1,\dots,u_d)\in A$ let $C^A_u $ be the polygonal chain connecting the following sequence of vertices: $(1,u_1-1),(2,u_2-1),\dots,(d,u_d-1)$.
        \item For each vector $v=(v_1,\dots,v_d)\in B$ let $C^B_v$ be the polygonal chain connecting the following sequence of vertices: $(1,1-v_1),(2,1-v_2),\dots,(d,1-v_d)$.
    \end{itemize}
    This construction takes time linear in the total length of the vectors in $A$ and $B$.
    Observe that $C^A_u$ intersects $C^B_v$ if and only if for some index $i\in [d]$ we have $u_i-1=1-v_i$, which can only happen when $u_i=v_i=1$. Consequently, $C^A_u$ intersects $C^B_v$ if and only if $u\in A$ and $v\in B$ are non-orthogonal.

    Thus $A,B$ has an orthogonal pair if and only if the constructed intersection graph is not a clique.
\end{proof}

\section{Diameter Lower Bound for $3$-slope Segments}

\begin{figure}[t]
    \centering
    \includegraphics[width=\linewidth]{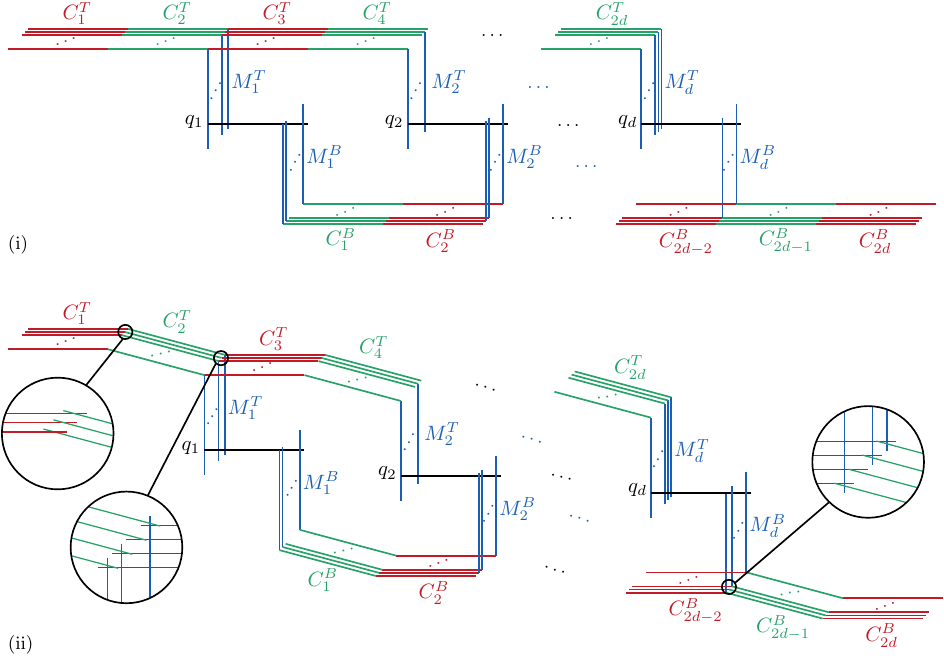}
    \caption{(i) The construction of~\cite{Bringmann2022-me} for axis-aligned segments. (ii) We introduce a small rotation on the green segment sets (those of even index in the top row and odd index in the bottom row). The resulting intersection graph is identical to the one in~\cite{Bringmann2022-me}, avoids degeneracies, and uses $3$ different slopes.}
    \label{fig:segment3slope}
\end{figure}

\begin{theorem}[Corollary of~\cite{Bringmann2022-me}]\label{thm:segment_3lope_lower}
Assuming the OV hypothesis,  there is no $O(n^{2-\eps})$ time algorithm for deciding if the intersection graph of a given set of $n$ segments with $3$ different slopes in $\reals^2$ has diameter at most $\Delta$ where $\Delta=O(\log n)$. 
\end{theorem}

\begin{proof}
    This is a simple consequence of a reduction of Bringmann et al.~\cite{Bringmann2022-me} for axis-aligned degenerate segments.  Bringmann et al.~\cite{Bringmann2022-me} construct a set of axis-aligned segments that are grouped into several smaller independent sets, as depicted in \Cref{fig:segment3slope}. We apply an affine transformation on the groups $C^T_i$ for even $i$ and $C^B_j$ for odd $j$ that changes the slopes of the affected segments but keeps their length as unit. By applying some infinitesimal translations on the segment groups so that triple intersections are replaced with triangles, the resulting intersection graph remains an intersection graph of unit segments with $3$ slopes without degeneracies.
\end{proof}
\section{\Diameter-2 Algorithm for Unit-disk Graphs}
\label{sec:diam2}
The goal of this section is to prove the following theorem, which is a more precise re-statement of the second part of \Cref{thm:diameter2-unitdis}.

\begin{theorem}\label{thm:udgdiam2}
$\Diameter$-2 for the intersection graphs of $n$ unit disks can be solved in time $O(n^{4/3}\log^4 n)$.
\end{theorem}

\subsection{Flowers and Cells}

Let $P$ be a set of points on the plane. Let $G$ be the undirected, unweighted unit-disk graph constructed from $P$. That is, $V(G) = P$, and there exists an edge $(x,y)$ between two points $x$ and $y$ if and only if $\|x,y\|\leq 1$. 
Our goal is to decide if $G$ has a diameter $2$ or its diameter is at least $3$. This problem can be re-phrased as a disk covering problem as follows.

\paragraph{Rephrasing as disk covering problem.} 
Let \EMPH{$\cD$} be the set of unit radius disks on the plane centered at points in~$P$. 
Observe that there is an edge $(x,y)$  in $G$ if and only if the disk centered at $x$ contains $y$ (and vice versa).  For each point $p$ in $P$, let \EMPH{$D_p$} be the disk centered at $p$ and $\EMPH{$\cD_p$} \coloneqq \set{D \in \cD : p \in D}$, which is the set of disks containing $p$.  Observe that centers of the disks in $\cD_p$ would be exactly $p$ and the \emph{1-neighborhood} of $p$ in $G$. We denote by $F_p$ the union of all disks in $\cD_p$, i.e., $\EMPH{$F_p$} \coloneqq \bigcup_{D \in \cD_p} D$; $F_p$ is a geometric region in the plane, called the \EMPH{flower of $p$}. See Figure~\ref{fig:flower}. In the following discussion, we use 
\EMPH{$\| p, q\|$} as the Euclidean distance between $p, q$ and \EMPH{$d_G(p, q)$} as the shortest path length between $p, q$ in the unit disk graph $G$.

\begin{figure}
\centering
\includegraphics[width=0.25\linewidth]{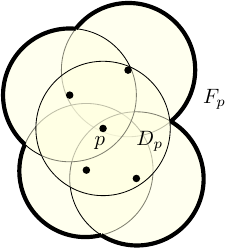}
\caption{The flower $F_p$ of a disk $D_p$ centered at $p$.}
\label{fig:flower}
\end{figure}

\begin{observation}\label{obs:dist-2} $d_G(p,q)\leq 2$ if and only if $q \in F_p$ (and hence $p \in F_q$).
\end{observation}

\begin{proof}
    Since $d_G(p,q)\leq 2$, then there must be a point $x\in P$ such that $\| x,p\|\leq 1$ and $\| x,q\|\leq 1$.  Therefore, both $p$ and $q$ are inside the disk $D_x$. Also $D_x\in \cD_p$.  Thus, $q\in F_p$. By symmetry, $p \in F_q$ as well.
\end{proof}
By Observation~\ref{obs:dist-2}, for $G$ to have diameter $2$, $F_p$ has to cover all points in $P$ for every single $p$; in other words, we want to check if 
\[
P \subseteq \bigcap_{p\in P} F_p.
\]

\begin{remark}
    The above formulation makes it clear that our problem is very similar to the Discrete $2$-center problem as formulated in~\cite{ASW98}. In~\cite{ASW98} the authors need to consider intersections of disk sets $K_p=\bigcap \bar{\mathcal  D_p}$, and then the goal is to compute the union $\bigcup_{p\in P} K_p$; there the set $\bar{\mathcal D_p}$ consists of the disks not containing $p$. This is in a sense ``dual'' to our setting. Consider stereographic projections of both problems onto a sphere: there the complement of each spherical disk is a spherical disk, and the complement of the intersection of the flowers is equal to the union of the intersections $K_p$. This transformation preserves the combinatorial properties and in particular the complexity of the boundary of the set $\bd \bigcap_p F_p$, however, it does not preserve the unit disk property, thus we cannot simply cast our problem as a geometric transformation of the problem in~\cite{ASW98}. However, we still benefit greatly from the techniques of~\cite{ASW98}, as we can use the same implicit representation for our flower intersections, and we can design an algorithm that mirrors the algorithm of~\cite{ASW98}.
\end{remark}

\paragraph{Griding.~}   Impose a grid with side-length $\delta$ which is a tiny constant less than $1$; its precise value will be determined later. Note that we can assume that points in $P$ are contained in a square of side length $3$ --- otherwise, the diameter of $P$ will be at least $3$. Thus, there are $O(1)$ grid cells that contain points of $P$.  For every pair of cells $A$ and $B$ that have a non-empty intersection with $P$, we want to check if all points in $A$ and all points in $B$ are within $2$ hops of each other in $G$.

We distinguish three cases. Denote by $\dist(A,B)$ the smallest distance between any two points in the cells $A$ and $B$ respectively. (1) If $\dist(A,B)\leq 1-2\sqrt{2}\delta$,  all points in $A\cap P$ are directly connected to all points in $B\cap P$. (2)
If $\dist(A,B)> 2$, there will be no single unit disk that intersects both $A$ and $B$ and hence there is no path of hop length at most 2 connecting a point in $A$ and a point in $B$ in $G$. Therefore, for the remaining discussion we consider pairs of cells $A,B$ such that $1-2\sqrt{2}\delta\leq \dist(A,B)\leq 2$.

Let $c_A$ and $c_B$ denote the centers of the cells $A$ and $B$, respectively. The \EMPH{origin $o$} of the cell pair $A,B$ is defined as the point on the segment $c_Ac_B$ at distance $1/3$ from $c_A$.

We will need the following lemma.

\begin{lemma}\label{lem:ostab}
     Let $A,B$ be a pair of cells of side length $\delta$ such that $1-2\sqrt{2}\delta\leq \dist(A,B)\leq 2$. Then any unit-radius disk that intersects both $A$ and $B$ is stabbed by the origin $o$.
\end{lemma}

\begin{proof}
    Consider circumscribed disks $D_A, D_B$  for $A$ and $B$, and a unit radius disk $D$ touching both $D_A, D_B$. Notice that any disk $D'$ intersecting both $A$ and $B$ will also intersect $D_A$ and $D_B$. Moreover, $D\cap c_Ac_B\subseteq D'\cap c_Ac_B$.
    See Figure~\ref{fig:stab}. 
    Let $\ell:=\dist(c_A,c_B)$ and recall that $1-\sqrt{2}\delta \leq \ell \leq 2+\sqrt{2}\delta$. We will now compute the distance of $c_A$ and the closer intersection point $q$ of the boundary of $D$ with $c_Ac_B$, which is 
    \[
    \frac{\ell}{2}-\sqrt{1-\left(1+\frac{\delta}{\sqrt{2}}\right )^2+\left (\frac{\ell}{2} \right)^2}=\frac{\ell}{2}-\sqrt{\left (\frac{\ell}{2} \right)^2-\left(\frac{\delta^2}{2}+\sqrt{2}\delta\right)}\leq \sqrt{\frac{\delta^2}{2}+\sqrt{2}\delta}\leq \frac{1}{3}
    \]
    Thus, all unit disks that intersect both $A$ and $B$ also contain the point $o$. 
\end{proof}

\begin{figure}
\centering
\includegraphics[width=0.5\linewidth]{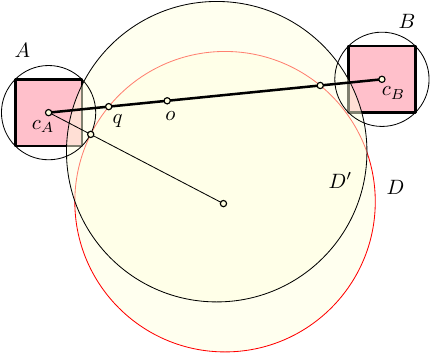}
\caption{Any unit radius disk that intersects both $A$ and $B$ is stabbed by $o$.}
\label{fig:stab}
\end{figure}

Now, to decide if every point in cell $A$ can be connected to every point in cell $B$ by a path of hop length at most $2$, it suffices to consider a subset of disks \EMPH{$\cD'$} in $\cD$ that intersect both cell $A$ and $B$, because only such disks have the potential to help to connect points in $A$ with points in $B$ in at most two hops. Take a point $p$ in $A$, let $\cD'_p$ be all disks in $\cD'$ that contain $p$, and let \EMPH{$F'_p$} be the union of the disks in $\cD'_p$. By \Cref{lem:ostab} we have that all disks in $\cD'$ are stabbed by the origin $o$. For the rest of this section we simply consider $\cD=\cD'$ as the base set of disks, and similarly we use $F_p$ to denote the flower corresponding to some $p\in A$. The following lemma gives us an algorithm to check hop distances between points in two cells.

\begin{lemma}
\label{lm:check-cells} 
Let $\delta>0$ be a sufficiently small constant. Let $A$ and $B$ be two squares of side length $\delta$ such that $1-2\sqrt{2}\delta \leq \dist(A,B) \leq 2$. Let $P$ be a set of points in $A\cup B$, and let $\cD$ be a set of disks on the plane such that every disk $D\in \cD$ intersects both $A$ and $B$. Set $n=|A|+|B|+|\cD|$.  For every point $p\in P$, let \EMPH{$F_p$} be the union of the disks in $\cD$ containing $p$.
Then in $O(n^{4/3}\log^4 n)$ time we can check whether or not
\[
P \cap B \subseteq \bigcap_{p \in P\cap A} F_p.
\]
\end{lemma}

\Cref{thm:udgdiam2} is a straightforward corollary of \Cref{lm:check-cells}:

\begin{proof}[Proof of~\Cref{thm:udgdiam2}]
    By applying the subroutine in \Cref{lm:check-cells} to every pair of cells, we can correctly decide if $G$ has a hop diameter at most 2 or not. Since there are only $O(1)$ pairs of cells for a constant $\delta$ and for each case the number of points and disks used by the subroutine is $O(n)$, the total running time of the algorithm $O(n^{4/3}\log^4 n)$.
\end{proof}

The rest of this section will focus on proving \Cref{lm:check-cells}.

\subsection{Flowers as Weak Pseudodisks}\label{sec:2cell}

Let $F_1$ and $F_2$ be a pair of regions whose boundaries are Jordan curves\footnote{A Jordan curve is a simple closed curve in the plane.}. 
An \EMPH{overlap $\sigma$} of $\bdry F_1$ and $\bdry F_2$ is a connected component of $\bdry F_1\cap \bdry F_2$, which is a Jordan arc. Then there is some small open neighborhood $N_\sigma$ of the overlap whose closure is disjoint from all other overlaps of $\bdry F_1$ and $\bdry F_2$. 
We say that the overlap $\sigma$ is a \EMPH{crossing overlap} if the intersections of $\bdry N_\sigma$ with $\bdry F_1$ and $\bdry F_2$ alternate along the walk around $\bdry N_\sigma$.
An overlap that is not a crossing overlap is called a \EMPH{ghost overlap}.
Notice that an overlap is a Jordan arc that has two well-defined endpoints, and they may coincide if the overlap degenerates to a single point. 

A set $\mathcal{S}$ of closed Jordan regions are called a \EMPH{weak pseudodisk} family if for any pair of sets $S,S'\in \mathcal S$ there are at most two crossing overlaps between $\bd S$ and $\bd S'$.

\begin{lemma}\label{lm:pseudo-disks}$\{F_p\}_{p\in P}$ are a weak pseudodisk family.
\end{lemma}

\begin{proof}
Given $p, q\in P$, let $F^*_p$ denote the union of disks containing $p$ but not containing $q$, and similarly let $F^*_q$ denote the union of disks containing $q$ but not containing $p$. Finally, let $F_{pq}$ denote the union of disks containing both $p$ and $q$. Now we can write $F_p=F^*_p\cup F_{pq}$ and $F_q=F^*_q\cup F_{pq}$. Observe now that $F^*_p,F^*_q$ cannot have any overlapping arc on their boundaries that is not a single intersection point, as the corresponding disk would necessarily have to contain both $p$ and $q$. However, the two boundaries can \emph{touch} in a single point, i.e., have a ghost overlap of length $0$, as an intersection point of two consecutive arcs from $\bd F^*_p$ can be on $\bd F^*_q$.

Our goal in this lemma will be to show that the family $\{F^*_p,F^*_q\}$ is a weak pseudodisk family. Before that, we will first prove that this implies that $\{F_p, F_q\}$ is also a weak pseudodisk family.

Consider first a pair $F,F'$ that form a weak pseudodisk family whose boundary consists of circular arcs that are both star-shaped from $o$, by Lemma~\ref{lem:ostab}. Add the same disk $D$ to both $F$ and $F'$ in such a way that $F\cup D$ and $F'\cup D$ remain star-shaped from $o$. Then we claim that the weak pseudodisk property is maintained, i.e., $F\cup D$ and $F'\cup D$ also form a weak pseudodisk family. Indeed, consider the labeling of each arc $\gamma$ of $\bd(F\cup F'\cup D)$ with $F,F'$ or $D$, based on if $\gamma$ belongs to $F,F'$ or $D$, respectively. We can also do a similar labeling on $\bd(F\cup F')$. Notice that if the arc $\gamma\in \bd(F\cup F'\cup D)$ is labeled as $D$, then it is a crossing overlap of $\partial (F\cup D)$ and $\partial (F'\cup D)$ if and only if the arc just before and just after $\gamma$ along $\bd(F\cup F'\cup D)$ are labeled with $F$ and $F'$ respectively, or vice versa. Altogether, the number of crossing overlaps between $F\cup D$ and $F'\cup D$ on $\bd(F\cup F'\cup D)$ is equal to the number of alternations between the labels $F$ and $F'$ along $\bd(F\cup F'\cup D)$, excluding all arcs labeled $D$. However, every arc labelled by $F$  and $F'$ in $\bd(F\cup F'\cup D)$ must also appear on $\bd F\cup F'$ with the same label (potentially as part of a longer arc), and thus any alternation between labels $F$ and $F'$ means a corresponding alternation along $\bd(F\cup F')$. Thus if $\bd(F\cup F'\cup D)$ has at least three crossing overlaps, then both $\bd(F\cup F'\cup D)$ and $\bd(F\cup F')$ has at least three alternations, thus $\bd(F\cup F')$ has at least three crossings, which contradicts the fact that $F,F'$ are pseudodisks.

Using the above claim repeatedly, we can add each disk covering both $p$ and $q$ to both $F^*_p$ and $F^*_q$, and maintain the pseudodisk property at each step, until we arrive at $F_p$ and $F_q$.

It remains to prove that $F^*_p,F^*_q$ are weak pseudodisks.

Consider a ray $R$ which starts from $p$ and follows the line through $p,q$ towards $q$. Define $x_p$ ($x_q$) as the intersection of $R$ with the boundary of $F^*_p$ ($F^*_q$). (See Figure~\ref{fig:crossing} for an example.) Note that $x_p$ is strictly closer to $p$ than $x_q$: indeed, each disk contributing to $F^*_p$ intersects the line through $p,q$ in an interval that contains $p$ and does not contain $q$, and the analogous claim holds for $F^*_q$ and $q$. Thus, the order of the points along $R$ is $p,x_p,q,x_q$. 
By symmetry, along the ray $R'$ starting from $q$ with the opposite direction of $R$, the intersection with $\bdry F^*_q$ (denoted as $x'_q$) is no further away from $q$ than the intersection of $R'$ with $\bdry F^*_p$ (denoted as $x'_p$). 

Now we walk along the boundary of $F^*_p$ and $F^*_q$ from $x_p, x_q$ to $x'_p$, $x'_q$ counterclockwise and examine the relative order of the boundary and crossings. We argue that the ordering of $\bdry F^*_p$ and $\bdry F^*_q$ changes strictly at most once. If not, by parity, the ordering has to change strictly at least twice. This means that there are three crossings $w_1, w_2, w_3$ in the counterclockwise order. 

\begin{figure}
\centering
\includegraphics[width=0.5\linewidth]{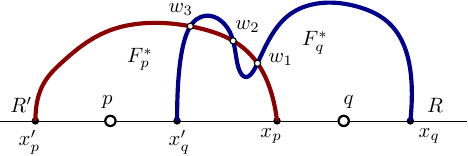}
\caption{The flowers $\{F^*_p\}_{p\in P}$ are weak pseudo disks.}
\label{fig:crossing}
\end{figure}

Now between crossing overlaps $w_1$ and $w_2$, there is a portion of the boundary of $F^*_p$ which stays strictly outside of that of $F^*_q$. That means that there is a disk $D_1$ whose boundary is along that portion of $\bdry F^*_p$; and $D_1$ includes $p$ but not $q$.  Similarly, there is a disk $D_2$ whose boundary appears on $\bdry F^*_q$ between $w_2$ and $w_3$; and $D_2$ includes $q$ but not $p$. 
We argue that $\bdry D_1$, $\bdry D_2$ intersect at least four times: twice as $D_1$ is outside of $D_2$ between $w_1$ and $w_2$ and $D_2$ outside of $D_1$ between $w_2$ and $w_3$; and twice since $D_1$ needs to include $p$ but not $q$ and $D_2$ needs to include $q$ but not $p$, so in particular the point of $\bd D_1$ directly below the line though $p,q$ is not in $D_2$ and analogously the point of $\bd D_2$ directly below the line through $p,q$ is in $D_1$. This is a contradiction as $D_1$ and $D_2$ are disks.
\end{proof}

\begin{remark}
   The above proof works for flowers defined as the union of convex, weak pseudodisks stabbed by the same point. 
\end{remark}

For a point $p$, the flower $F_p$ is not necessarily convex. Thus if we consider a point $y$ inside $F_p$ and shoot a ray from $y$, it may intersect $\bdry F_p$ multiple times.
But if $y$ is sufficiently close to the point $p$, there is only a unique intersection. 
\begin{lemma}\label{lem:ray-intersection}
    Let $R$ be a ray originated from some point $y$ inside a flower $F_p$ with $\|y,p\|< 1/2$. Then $R$ intersects $\bdry F_p$ at a single point. 
    Consequently, $\bd F_p$ is a Jordan curve with $y$ in the interior of its bounded region.
\end{lemma}
\begin{proof}
    Assume otherwise, that $R$ intersects  $\bdry F_p$ at more than one points. Since $y$ is inside $D_p$, $y$ is  inside $F_p$ which includes the disk $D_p$.  Thus $R$ intersects $\bdry F_p$ at at least three points, such that $R$ exits and enters $F_p$ in an alternating way. All the intersections are outside $D_p$.

    If there are a pair of  intersections $t$ and $t'$ on the boundary of the same disk $D_s$. $D_s$ has $p$ inside, and $t, t'$ are both outside $D_p$. Thus the position for $y$ with the largest possible $\|y,p\|$ is realized when the ray $R$ is (nearly) tangent to $D_s$ at the point $w$, an intersection of $\bdry D_p$ and $\bdry D_s$. See \Cref{fig:ray-intersection}. In this case $\|y,p\|=1/2$. Thus when $\|y,p\|<1/2$, this is not possible. 

    \begin{figure}
    \centering
    \includegraphics{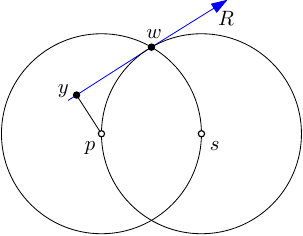}
    \caption{The ray $R$ starting from a point $y$ with distance within $1/2$ from the center $p$ intersects $\bdry F_p$ only once.}
    \label{fig:ray-intersection}
    \end{figure}
    
    Now consider a pair of intersections $t_1$ and $t_2$ on the ray $R$, where we enter $F_p$ at $t_1$ and exit $F_p$ at $t_2$. Suppose that the two points stay on the boundary of two disks centered at $a$ and $b$ respectively. Now suppose $R$ enters the disk $D_{b}$ at point $t'_2$ (and exits the disk at point $t_2$). Between $t'_2$ and $t_2$ the ray $R$ stays inside the disk $D_{b}$. We claim that $t'_2$ cannot come before $t_1$ on the ray $R$ -- if otherwise, this gives a contradiction that $R$ enters $F_p$ at $t_1$ since it would have already entered $F_p$ at point $t'_2$. Now we argue that $t'_2$ is outside of $D_p$. This is because that $t_1$ is already outside $D_p$ since we (re)-enter $F_p$ at point $t_1$. And for a ray $R$ once it exits $D_p$ it stays outside of $D_p$. Now we have on the ray $R$ two points $t'_2$ and $t_2$, both outside of $D_p$, staying on the boundary of the same disk $D_{b}$ which includes $p$ inside. In order to create such a ray, the point $y$ cannot be closer than $1/2$ from $p$ by the same argument as in \Cref{fig:ray-intersection}.
\end{proof}


\begin{lemma}\label{lm:orientation} 
Let $p,q$ be points with $\|p,q\|< 1/2$, and let $\ell_{pq}$ be the perpendicular bisector of $pq$, defining open half-planes $H_p\ni p$ and $H_q \ni q$. Let $U$ denote one of the half-planes defined by the line $pq$.
Suppose that in the quadrant $H_q\cap U$ there is some point $x\in \bd F_p\cap \bd F_q$, and let $C_x$ denote the cone with boundary lines $qx$ and $pq$ contained in $H_q\cap U$. Then $\bd(F_p\cap F_q)\cap C_x=\bd F_p\cap C_x$, i.e., the boundary of $F_p\cap F_q$ between $x$ and the line $pq$ is the same as $\bd F_p$.
\end{lemma}

\begin{figure}
    \centering
    \includegraphics{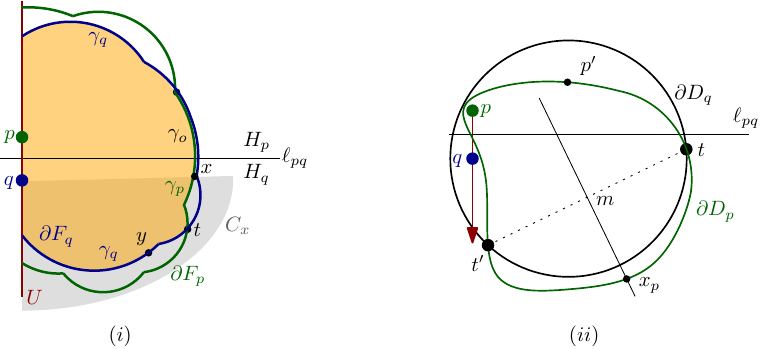}
    \caption{(i) The situation ruled out by \Cref{lm:orientation}. The flower intersection $F_p\cap F_q$ inside the half-plane $U$ is depicted as the orange shaded region. Its boundary can be decomposed to parts that are contributed only by $\bd F_p$, only by $\bd F_q$ (denoted by $\gamma_p$ and $\gamma_q$), or both (denoted by the curve $\gamma_o$). The lemma claims that a point $y\in \gamma_q$ cannot be below an overlap point $x\in \gamma_o\cap H_q$. (ii) The impossible configuration: the unit circles $\bd D_q$ and $\bd D_p$ with their shared chord $tt'$. The circle $\bd D_p$ is deformed to depict the correct order of intersections with $\ell_{pq}$, $\bd D_q$ and the line $pq$.}
    \label{fig:binarysearchfix}
\end{figure}

\begin{proof}
Assume without loss of generality that $\ell_{pq}$ is horizontal and $p$ is above $q$.
Let $\gamma_o := U\cap \bdry F_p \cap \bdry F_q$ denote the overlaps of $\bd F_p$ and $\bd F_q$ inside $U$, and set $\gamma_p:=U\cap \bdry (F_p \cap F_q) \setminus F_q$ as well as $\gamma_q:=U\cap  \bdry (F_p \cap F_q) \setminus F_p$ denoting the portions of $\bd (F_p\cap F_q)$ that are contributed only by $\bd F_p$ and only by $\bd F_q$, respectively.  Suppose that there is some point $x\in \gamma_o\cap H_q$. We will show that for any $y\in \gamma_q$, $\angle xqp> \angle yqp$. Observe that this implies that the boundary of $\bd(F_p\cap F_q)$ below $x$ is the same as the boundary of $\bd F_p$, concluding the lemma.

 \Cref{lem:ray-intersection} implies that the ray $\ell_{pq}\cap U$ has a unique intersection with $\bd F_p$ and $\bd F_q$, and it also implies that the line $\bd U$ intersects $\bd F_p$ exatcly twice and $\bd F_q$ exactly twice. For points $a,b\in U$ we say that $a$ is clockwise \emph{above} $b$ if $\angle aqp< \angle bqp$, and similarly $a$ is clockwise \emph{below} $b$ if $\angle aqp> \angle bqp$. Notice that by \Cref{lem:ray-intersection} the rays from $q$ have a unique intersection with both $\bd F_p$ and $\bd F_q$, thus above/below is a complete ordering on both $U\cap \bd F_p$ and $U\cap \bd F_q$. Our goal can be equivalently stated as follows: prove that if an overlap point $x\in \gamma_o$ falls in $H_q$, then there are no points of $\gamma_q$ below $x$.

Suppose for the sake of contradiction that there is some point $y\in \gamma_q$ below $x$, as in \Cref{fig:binarysearchfix}(i). Let $t$ be the lowest point of $\gamma_o$ above $y$. Notice that the arc of $\bd (F_p\cap F_q)$ entering $t$ from above is an arc of some disk $D_q\in \cD_q$ while there is some disk $D_p\in \cD_p\setminus \cD_q$ where the arc of $\bd D_p$ below $t$ is exiting $F_q$. For example, the disk $D_p$ contributing the boundary arc of $\bd F_p$ leaving from $t$ downward is a valid choice. See \Cref{fig:binarysearchfix}(ii) for an illustration. (We emphasize that $q\not \in D_p$, but we do allow $p \in D_q$.) Let $t'$ denote the other intersection of $\bdry D_{q}\cap \bdry D_p$. The ray from $p$ passing through $q$ exits $D_p$ first, then $D_q$ second (since $q\in D_q\setminus D_p$), thus $t'\in U\cap H_q$. Now $tt'$ is the unique shared chord of the unit circles $\bdry D_q$ and $\bdry D_p$, and the segment $tt'$ is in $H_q$. 

Let $x_p$ denote the intersection of the bisector of $tt'$ with $(\bd D_p)\setminus D_q$. We claim that $x_p\in H_q$. To see why, assume the contrary: that $x_p \not \in H_q$. Since $p\in D_p$ is in the open half-plane $H_p$, there is a point $p'\in H_p \cap \bd D_p$. Then walking along $\bd D_p\cap U$ top to bottom, we encounter the points $x_p t p' t'$ in this order, where $x_p,p'\in H_p$ and $t,t'\in H_q$. Thus $\bd D_p$ must intersect $\ell_{pq}$ at least three times (as $t,t'\in H_q$, $p\in H_p$ and $x_p\in H_p\cup \ell_{pq}$), which contradicts the fact that a circle can intersect a line at most twice.

Let $m$ denote the midpoint of $tt'$. Since $t,t'\in H_q$ we have that $m\in H_q$. Now $m\in H_q$ and $x_p \in H_q$ by our earlier claim. Observe that the center of $D_p$ is on the segment $mx_p$, thus the center of $D_p$ is in $H_q$. On the other hand we have that the center of $D_p$ is closer to $p$ than to $q$ (since $p\in D_p$ and $q\not\in D_p$), so it is in $H_p$. This is a contradiction, and concludes the proof.
\end{proof}

\subsection{Implicit Representation and Basic Primitives}

The main challenge of \Cref{lm:check-cells} is that the boundary of a flower $F_p$ might be composed of arcs on the boundary of $\Omega (n)$ disks and hence may have $\Omega(n)$  vertices. Since there could be $\Omega(n)$ points in $A$ (and $B$), representing all $\{F_p: p \in A\}$ explicitly requires $\Omega(n^2)$ space. Here, we work with an implicit representation of $\{F_p: p \in A\}$ by adapting the techniques in \cite{KS97}. First, we describe the implicit representation and basic primitives that we can apply to operate on the implicit representation.

\paragraph{Preprocessing.~} Following~\cite{ASW98}, we describe a slightly more general setting in which we are given a set $\cD$ of $n$  disks in the plane and a set of points $P$. Our goal is to represent for each point $p$ a set of disks in $\cD$ that contain $p$, i.e., the flower $F_p$.  

\begin{lemma}[\cite{KS97}, adapted]\label{lm:rep} Let $\cD$ be a set of $n$ disks on a plane and $P$ be a set of points. For every point $p\in P$, let $\cD_p$ be the subset of disks in $\cD$ that contain $p$. Then there exists a family  $\set{\cD^{(1)}, \dots, \cD^{(s)}}$ of \EMPH{canonical subsets} of $\cD$ such that $\sum_{i=1}^s |\cD^{(i)}| = O(n^{4/3} \log n)$, and such that, for any $p \in P$, $\cD_p$ can be represented as the union of $O(n^{1/3} \log n)$ canonical subsets. \\
Furthermore, for every $p\in P$, let $J_p$ be the set of indices of the canonical subsets representing $\cD_p$ (i.e., $\cD_p = \bigcup_{i \in J_p} \cD^{(i)}$). Then one can compute  $\set{J_p}_{p \in P}$ in time $O(n^{4/3} \log n)$.
\end{lemma}
 For each canonical subset $\cD^{(i)}$ where $i\in [1,s]$, we compute (the boundary of) $C_i = \cup_{D\in \cD^{(i)}}D$ in $O(|\cD^{(i)}|\log(|\cD^{(i)}|))$ time. Let $V(\bdry C_i)$ be the set of vertices on the boundary of $C_i$, i.e, the intersection point of two unit circle arcs on $\bd C_i$. Then:
 \begin{equation}\label{eq:Fi-size}
     \sum_{i=1}^{s}|V(\bdry C_i)| = \sum_{i=1}^{s} O(|\cD^{(i)}|) = O(n^{4/3}\log n )
 \end{equation}
The total preprocessing time is:
 \begin{equation}\label{eq:preprocessing}
    O(n^{4/3} \log n) + \sum_{i=1}^{s} O(|\cD^{(i)}|\log(|\cD^{(i)}|)) = O(n^{4/3}\log^2 n)
 \end{equation}
 
 Using the canonical representation, one can decompose the boundary curve $\bdry F_p$ into a set $\Gamma_p$ of $\Tilde{O}(n^{1/3})$ subcurves where each subcurve $\gamma\in \Gamma_p$ is a boundary segment of the union of disks in a canonical set. That is, $\gamma \subseteq \bdry (\cup_{D \in \cD^{(i)}} D)$ for some $i \in [1,s]$. See Figure~\ref{fig:canonical}. This is the key for us to manipulate $\bdry F_p$  in $\Tilde{O}(n^{1/3})$ time and ultimately achieve a truly subquadratic time to compute $\bigcap_{p} F_p$. Naively, the number of disks contributing to $\bdry F_p$   could be $\Omega(n)$, and without a compact representation, computing $\bigcap_{p} F_p$ could  take $\Omega(n^2)$ time. 

     \begin{figure}
    \centering
    \includegraphics{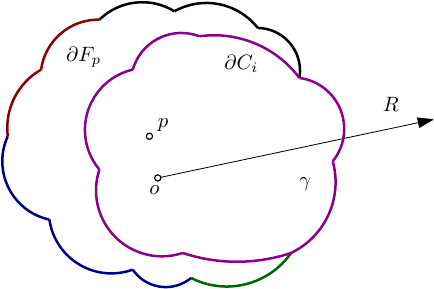}
    \caption{The boundary $\bdry F_p$ implicitly represented by curves from canonical sets (shown in different colors). The boundary of one canonical set $C_i$ is also shown.}
    \label{fig:canonical}
    \end{figure}

We work with the setup in \Cref{lm:check-cells} where $\cD$ is a set of disks that intersect both cells $A$ and $B$, and $F_p$ is the union of disks in $\cD$ that contain $p$. We build the following data structure $\cF$. First, we compute, for each canonical disk union $C_i$ the complete description of $\bd C_i$ by listing the disk indices that contribute to the boundary in the clockwise order of the boundary when viewed from $o$. Recall that by \Cref{lem:ray-intersection} each ray from the origin $o$ (defined with respect to the two cells $A, B$) has a unique intersection with $\bd C_i$, so we can also get an explicit list of $V(\bd C_i)$ in clockwise order by computing the intersection point of consecutive arcs. We store the list of arcs (with indices) in an array. The time required to compute $\cF$ is also $O(n^{4/3}\log^2 n)$. In the same amount of time we also store a large array $\Lambda$ of all vertices $\bigcup_i V(\bd C_i)$ sorted according to the cyclic order of rays from $o$. We also store pointers to the circular arcs preceding and succeeding each $v\in V(\bd C_i)$ together with their entry in $\Lambda$.

Next, we provide two basic primitives that we could apply to operate on the implicit representation of $F_p$. (To be more precise, we operate on the implicit representation of $\cD_p$.) Our basic primitives are inspired by that of~\cite{ASW98}.  

\begin{lemma}[Find the intersection point of a ray $R$ and $\bdry F_p$]\label{lm:ray-shooting} 
Let $R$ be a ray originated from some point $o$ inside a flower $F_p$ with $\|o,p\|<1/2$. 
We can construct a data structure in  $O(n^{4/3}\log^2 n)$ time so that given $R$ and $F_p$, we can query their intersection point as well as the circular arc(s) of $\bd F_p$ that are incident to the intersection point in $O(n^{1/3}\log^2 n)$ time.
\end{lemma}
\begin{proof}
    First by \Cref{lem:ray-intersection} there is only one intersection of $R$ and $\bdry F_p$.
    Given a ray $R$ originated from $o$, and $i \in [s]$, we can find the (unique) intersection point of $R$ and $\bdry C_i$ using a simple binary search in the data structure $\cF$. Then, to compute the intersection of a ray $R$ and $F_p$, we compute the intersections of $R$ with $\bd C_i$  for all $i \in J_p$ in time $O(|J_p|\log n)$, and then decide which point among them belongs to  $\bdry F_p$ in time $O(|J_P|)$: notice that the intersection that is farthest from $o$ is the intersection of $\bd F_p$ with $R$. Note that during the comparison of distances we can also keep track of the outermost arc(s) incident to the outermost intersection point, which will be the only one(s) contributing to $\bd F_p$ at $\bd F_p \cap R$. The total running time is $O(|J_p|\log n) = O(n^{1/3}\log^2 n)$, since $|J_p|$ is $O(n^{1/3}\log n)$.
\end{proof}

Next, we provide a primitive to find crossings of the boundaries of any two given flowers $\bdry F_p$ and $\bdry F_q$. By \Cref{lm:pseudo-disks}, there are at most two crossing overlaps. Our idea is to apply binary search by ray shooting. Specifically, we will shoot rays originating from $o$, starting from an axis parallel ray. 
At some point in the search, we have a ray $R$, and we will have to decide whether to shoot the next ray in the counterclockwise or clockwise order of $R$, depending on whether the intersection point we look for is to the left (counterclockwise) or right (clockwise) of $R$. If $R$ intersects $\bdry F_p$ and $\bdry F_q$ at two distinct points, it is not so hard to decide whether the intersection point is on the left or the right of $R$. However, the difficult case is when $\bdry F_p$ and $\bdry F_q$ share boundary; this could happen when $\cD_p$ and $\cD_q$ contain the same (one or more) disks. In this case, we rely on the following lemma to make the decision of going clockwise or counterclockwise.

Let $C_B$ denote the minimum closed cone of center $o$ enclosing the cell $B$.

     \begin{figure}
    \centering
    \includegraphics{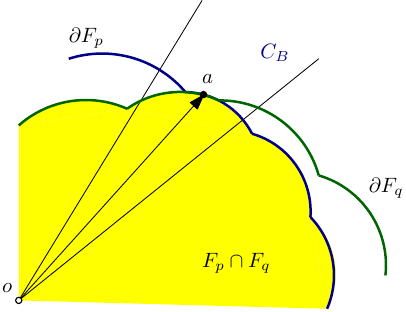} \hspace*{10mm}     \includegraphics{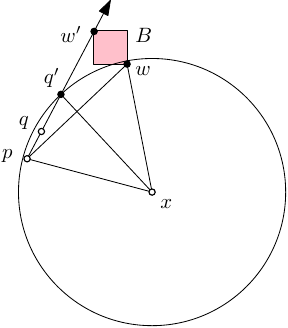}
    \caption{(Left) Crossing points of $\partial F_p$ and $\partial F_q$ inside cone $C_B$. (Right) $pq$ intersects cell $B$.}
    \label{fig:Fp-Fq}
    \end{figure}

\begin{lemma}[Crossing points of $\bdry F_p$ and $\bdry F_q$]
\label{lm:cross-point}
Given two points $p$ and $q$ in one grid cell $A$, either $\bd F_p\cap C_B$ or $\bd F_q\cap C_B$ is equal to the boundary $\bd (F_p\cap F_q)\cap C_B$, or there is some point $a\in \bd (F_p\cap F_q)$ that is a point of some crossing or ghost overlap of $\bd F_p$ and $\bd F_q$ such that $\bd (F_p\cap F_q)\cap C_B$ is the concatenation of one part of  $\bd F_p\cap C_B$ and one part of $\bd F_q\cap C_B$ separated by the ray $oa$. Furthermore, in $O(n^{1/3}\log^3 n)$ time, we can find the description of $\bd (F_p\cap F_q)\cap C_B$ with $\bd F_p$ and $\bd F_q$ and at most one breakpoint $a$.
\end{lemma}

\begin{proof}
    First, we show that a description of $\bd (F_p\cap F_q)\cap C_B$ as described by the lemma is possible.
    We distinguish two cases based on whether the line $pq$ intersects the cell $B$ or not.    
    
    In case (i), the line $pq$ intersects $B$. Assume without loss of generality that $q$ lies between $p$ and $B$ on the line $pq$. Then one can show that due to $\dist(p,q)\leq\sqrt{2}\delta$ and the position and size of $B$, any unit disk containing $p$ and intersecting $B$ also contains $q$, thus $\bd (F_p\cap F_q)\cap C_B=\bd F_p \cap C_B$. 
    To see why, suppose a unit radius disk $D_x$ centered at $x$ intersects both $p$ and cell $B$ (specifically, a point $w$ of $B$).     
    If the line $pq$ intersects a point $w'$ of $B$ which is inside the disk $D_x$, by convexity $q$ is also inside $D_x$ and we are done. Thus, we can assume wlog that $\|w, x\|\geq 1-\sqrt{2}\delta$ -- otherwise the entire cell $B$ is inside $D_x$. Further, we can assume wlog that $\|p, x\|\geq 1-\sqrt{2}\delta$ -- otherwise $q$ is also inside $D_x$. See the right hand side of Figure~\ref{fig:Fp-Fq}.
    Now we assume that one intersection $w'$ of line $pq$ with $B$ is completely outside of $D_x$. $\|p, w\|\geq 1-2\sqrt{2}\delta$, since the two cells $A, B$ have minimum distance at least $ 1-2\sqrt{2}\delta$. Also $\|w,w'\|\leq \sqrt{2}\delta$. Thus the angle $\angle wpw'$ is at most $\arcsin(\frac{\|w,w'\|}{\|w,p\|})\leq \arcsin(\frac{\sqrt{2}\delta}{1-2\sqrt{2}\delta})\approx \Theta(\delta)$. 
    Now we consider $\triangle pxw$. By cosine law, we have 
    \[
    \cos \angle wpx =\frac{\|p,x\|^2+\|p, w\|^2-\|x,w\|^2}{2\|p, x\|\cdot \|p, w\|}\geq \frac{(1-\sqrt{2}\delta)^2+(1-2\sqrt{2}\delta)^2-1}{2(2+2\sqrt{2}\delta)}\geq \frac{1}{4}-\frac{7\delta}{4(1+\sqrt{2}\delta)}
    \] 
    Thus $\angle wpx \leq 76^{\circ}+\Theta(\delta)$. 
    Define the intersection of the ray $pq$ with the boundary of $D_x$ as $q'$. $\angle q'px=\angle q'pw+\angle wpx$. $\|q',x\|=1\geq \|p,x\|$.  By sine law $\angle pq'x\leq \angle q'px$. And thus $\angle pxq'\geq 28^{\circ}-\Theta(\delta)$. If we choose $\delta$ to be a sufficiently small constant, we can make $\|p,q'\|\geq \sqrt{2}\delta$, which means that $q$ is inside $D_x$. This finishes the claim.

     \begin{figure}
    \centering
    \includegraphics[scale=1.1]{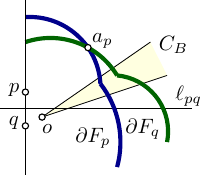} \hspace*{4mm}
    \includegraphics[scale=1.1]{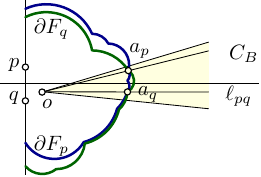}

    \caption{ (Left) $a_p$ and $a_q$ are not inside $C_B$. (Right) $a_p$ and $a_q$ are both inside $C_B$. }
    \label{fig:pqB-2}
    \end{figure}    

    In case (ii), the line $pq$ is disjoint from $B$. Then $\bd (F_p\cap F_q)\cap C_B$ lies entirely in one half-plane $U$ given by $pq$, which matches the setup in \Cref{lm:orientation}. Without loss of generality we assume that the line $pq$ is vertical with $p$ above $q$ and $U$ on the left hand side of line $pq$ (as in \Cref{fig:binarysearchfix} (i)). Let $\ell_{pq}$ be the perpendicular bisector of $p$ and $q$. Let $a_p$ denote the intersection point of $\bd F_p$ and $\bd F_q$, on the same side of $\ell_{pq}$ as $p$, that is closest to $\ell_{pq}$ in the angular order along the boundary $\bd (F_p\cap F_q)\cap C_B$. Similarly, denote by $a_q$ as the closest intersection on the same side of $q$. Notice that $a_p,a_q$ or both may not exist or not be part of $C_B$. If neither exists or falls into $C_B$ (see \Cref{fig:pqB-2} (Left)), then $\bd (F_p\cap F_q)\cap C_B$ is equal to one of  $\bd F_p\cap C_B$ or $\bd F_q\cap C_B$. If only $a_p$ falls in $C_B$, then $\bd (F_p\cap F_q)\cap C_B$ can be described in the claimed fashion with $a=a_p$. If both fall in $C_B$, $\ell_{pq}$ intersects $C_B$. See \Cref{fig:pqB-2} (Right).
    Let $C_{a_p}$, $C_{pq}$ and $C_{a_q}$ denote the decomposition of $C_B$ into three consecutive cones by rays $oa_p$ and $oa_q$, where $C_{pq}$ contains the intersection of $\ell_{pq}$ and $\bd(F_p\cap F_q)$ and contains $a_p$ and $a_q$ on its boundary, while $C_{a_p}$ and $C_{a_q}$ denote the two cones of $C_B\setminus C_{pq}$ on either side of $C_{pq}$, containing $a_p$ and $a_q$ on their boundary, respectively. By \Cref{lm:orientation} we have that the portion of $\bd(F_p\cap F_q)$ beyond $a_q$ (i.e., inside $C_{a_q}$) is equal to $\bd F_p$; similarly, $\bd(F_p\cap F_q)\cap C_{a_p}= \bd F_q \cap C_{a_p}$. Notice that in the interior of $C_{pq}$ the boundaries $\bd F_q$ and $\bd F_p$ are disjoint, thus the one closer to $o$ will be equal to $\bd(F_p\cap F_q)\cap C_{pq}$. Suppose that $\bd(F_p\cap F_q)\cap C_{pq}= \bd F_p\cap C_{pq}$ (the other case is symmetric). Then we can use $a:=a_p$ as the breakpoint, as $\bd(F_p\cap F_q)$ and $\bd F_q$ are equal in $C_{a_p}\cup C_{pq}$.

    It remains to compute the desired description of $\bd (F_p\cap F_q)\cap C_B$ in $\Tilde{O}(n^{1/3})$ time. We will do this via a binary search procedure that uses \Cref{lm:ray-shooting} to find the desired description. However, before the binary search we first check if the line $pq$ intersects $B$. If it does, then we conclude that $\bd F_p\cap C_B=\bd (F_p\cap F_q)\cap C_B$ if the ray $pq$ intersects $B$ and we conclude $\bd F_q\cap C_B=\bd (F_p\cap F_q)\cap C_B$ if the ray $qp$ intersects $B$.

    Suppose now that the line $\overleftrightarrow{pq}$ is disjoint from $B$ and thus $B$ lies in an open half-plane $U$ whose boundary is the line $\overleftrightarrow{pq}$. Assume without loss of generality that $\bd U$ is vertical, $p$ is above $q$, and $U$ is to the right of its boundary. By \Cref{lem:ray-intersection} we have that for any point $o'$ within distance $1/2$ from both $p$ and $q$ the rays emanating from $o'$ intersect $\bd F_p$ (and $\bd F_q$) in the same clockwise order. In particular, we can talk about a point of $\bd F_p$ being above/below another point based on their clockwise order along $\bd F_p$, and the same holds for $\bd F_q$ and even for $\bd (F_p\cap F_q)$.

    Using \Cref{lm:ray-shooting} twice (for $F_p$ and $F_q$) and the ray $R=\ell_{pq}\cap U$ we determine the intersections of $\bd F_p$ and $\bd F_q$ with $R$. Let $z$ denote the intersection closer to $\overleftrightarrow{pq}$, and assume without loss of generality that $z\in \bd F_q$.
    Let $R^\uparrow$, $R^\downarrow$ denote the boundary rays of $C_B$ where $R^\uparrow\cap \bd F_p$ is above $R^\downarrow \cap \bd F_p$. 
    If $z$ is below $R^\downarrow\cap \bd F_p$, then either there is no overlapping point of $\bd F_p$ and $\bd F_q$ between the part of $\bd F_q$ from $z$ to $R^\uparrow \cap \bd F_q$, and thus $\bd F_p\cap \bd F_q$ coincides with $\bd F_q$ on this range (and thus, also in $C_B$), or there is some nearest overlap point $x$, above which $\bd F_q$ and $\bd(F_p\cap F_q)$ coincide by \Cref{lm:orientation}. In both cases we can conclude that $\bd (F_p\cap F_q)\cap C_B=\bd F_q \cap C_B$ and terminate.
    
    If $z$ is above $R^\downarrow\cap \bd F_p$, then we find the nearest overlap of $\bd F_p$ and $\bd F_q$ below $z$ via binary search as follows. If $z\in \bd F_p\cap\bd F_q$, then the binary search can conclude as $z$ is a correct overlap point. Otherwise, we do binary search using the description of $\bd F_p$ in the data structure $\cF$ constructed earlier. First, we find the arc containing the intersection point $y^\downarrow:=\bd F_p\cap \overrightarrow{pq}$. By our earlier argument, we have that $y^\downarrow\in \bd(F_p\cap F_q)$. Similarly, we find the element of the array representing $\bd F_p$ that contains the point $y^\uparrow:=\bd F_p\cap \overrightarrow{oz}$. Notice that there must be some point $a\in \bd F_p\cap \bd F_q$ between $y^\uparrow$ and $y^\downarrow$ such that $\bd (F_p\cap F_q)$ is equal to $\bd F_q$ above $a$ and equal to $\bd F_p$ below $a$. We will maintain this invariant throughout the search.
    
    We do binary search using the array $\Lambda$: we select the point $y\in \Lambda$ midway between the nearest points before $y^\uparrow$ and the nearest point after $y^\downarrow$ (where after/before refers to the clockwise sorting of rays around $o$, which aligns with the sorting of $\Lambda$). Using \Cref{lm:ray-shooting} we find the points $\overrightarrow{oy}\cap \bd F_p$ and $\overrightarrow{oy}\cap \bd F_q$.
    If $\overrightarrow{oy}\cap \bd F_p$ and $\overrightarrow{oy}\cap \bd F_q$ coincide or if $\overrightarrow{oy}\cap \bd F_q$ is closer to $o$ than $\overrightarrow{oy}\cap \bd F_p$, then by \Cref{lm:orientation} we have that $\bd F_p$ and $\bd(F_p\cap F_q)$ are equal below $\overrightarrow{oy}$. We set $y\downarrow:=y$ and continue on the subarray of $\lambda$ between $y^\downarrow$ and $y^\uparrow$.
    If $\overrightarrow{oy}\cap \bd F_q$ is farther from $o$ than $\overrightarrow{oy}\cap \bd F_p$, then by \Cref{lm:orientation} we have that $\bd F_q$ and $\bd(F_p\cap F_q)$ are equal above $\overrightarrow{oy}$. We set $y\uparrow:=y$ and continue on the subarray of arcs between $y^\downarrow$ and $y^\uparrow$. Observe that in both cases we maintain the invariant.

    The search terminates when the subarray consists of two consecutive intersection points $y^\downarrow,y^\uparrow$ of~$\Lambda$. Notice however that the boundaries of $\bd F_p$ and $\bd F_q$ may still have several arcs each between $y^\downarrow
    $ and $y^\uparrow$ as $\Lambda$ does not contain intersection points between different canonical sets $\bd C_i, \bd C_j$ even if $C_i,C_j$ both come from the description of $F_p$. To deal with this issue, observe that for each $C_i$ contributing to $\bd F_p$ we have that there is only a single circular arc of $C_i$ intersecting the cone between $\overrightarrow{oy}^\downarrow$ and $\overrightarrow{oy}^\uparrow$, as the cone cannot contain any vertex of $V(\bd C_i)$. We can find all of these $|J_p|$ arcs via a binary search on each $\bd C_i$. Since $|J_p|=O(n^{1/3}\log n)$, the boundary of the union of these arcs can be computed in $O(|J_p|\log|J_p|)=O(n^{1/3}\log^2 n)$ time. Similarly, we compute the boundary of the arcs of each $C_j$ appearing in $J_q$ between $\overrightarrow{oy}^\downarrow$ and $\overrightarrow{oy}^\uparrow$. Finally, we sweep both of these boundaries in $O(n^{1/3}\log n)$ time to find the desired point $a$.

    To return the representation, we simply check if $a\in C_B$ and return the correct description of $\bd (F_p\cap F_q)$ in $C_B$ accordingly. The binary search procedure uses $O(\log n)$ calls to the subroutine of \Cref{lm:ray-shooting} thus it runs in $O((|J_p|+|J_q|)\log^2 n)=O(n^{1/3}\log^3 n)$ time, and this dominates the running time of the latter parts of the procedure.
\end{proof}

\begin{remark}
    We note that the algorithm presented for the analogous boundary-intersection-finding problem in~\cite{ASW98} (the subroutine matching \Cref{lm:cross-point}) has a small gap. Just as in our setting, the sets $K_p$ form weak pseudodisks and the goal is to compute a description of the boundary of a pairwise union $\bd(K_p\cup K_q)$. However, the existence of ghost overlaps means that a simple binary search using the ray-shooting subroutine is not possible. In our setting, \Cref{lm:orientation} allowed us to resolve the case when the ray of the binary search intersects both $\bd F_p$ and $\bd F_q$ in the same point. Presumably, a similar geometric lemma (analogous to \Cref{lm:orientation}) could fill this small gap in~\cite{ASW98}.
\end{remark}

\subsection{The Algorithm}\label{sec:untidiskalgo}

We show how to compute the boundary curve $\gamma = C_B\cap \bdry(\cap_{p\in A\cap P}F_p)$ in $\Tilde{O}(n^{4/3})$ time. The algorithm is by divide and conquer: divide the set $A\cap P$ into two subsets $A_1$ and $A_2$ of roughly equal size, compute the boundary $\gamma_1 = B\cap \bdry(\cap_{p\in A_1}F_p)$ and $\gamma_2 = B\cap \bdry(\cap_{p\in A_2}F_p)$. Then, we will compute $\gamma$ from $\gamma_1$ and $\gamma_2$ by sweeping a ray in $C_B$.

To implement the above algorithm efficiently, we have to have a compact presentation of $\gamma_1$ and $\gamma_2$. We rely on the following observation, which is a consequence of the fact that weak pseudodisks have linear union complexity.

\begin{observation}\label{obs:flower-int-rep}
Let $S$ be a subset of $P\cap A$. Let $\gamma_S =  C_B\cap \bdry(\cap_{p\in S}F_p)$. Then we can decompose $\gamma_S$ into at most $O(|S|)$ curves where each curve, denoted by $\alpha_p$, is a sub-curve of $\bdry F_p$ for some point $p\in S$.
\end{observation} 

Observe that we can (implicitly) represent a subcurve of $\bdry F_p$ of a point $p$  by a triple $(p,a, b)$ where $a$ and $b$ are the two endpoints of the subcurve, and the subcurve is the clockwise portion of $\bd F_p$ from $a$ to $b$. By Observation~\ref{obs:flower-int-rep}, we can represent $\gamma_S$ for any subset $S$ of  $P\cap A$ in $O(|S|)$ space by a linked list: each element of a link list contains a triple $(p,a,b)$ representing a sub-curve of $S$ on  $\bdry F_p$, such that if  $(p,a,b)$  and $(\hat{p},\hat{a},\hat{b})$ are two consecutive elements (representing two consecutive subcurves) in the linked list, then $b = \hat{a}$. For a curve $\gamma_S$ let $LL(\gamma_S)$ denote the corresponding implicit representation.

\begin{lemma}\label{lem:raysweep}
    Let $\gamma_1,\gamma_2,\gamma$ be as stated above (see beginning of \Cref{sec:untidiskalgo}). Then, given $LL(\gamma_1)$ and $LL(\gamma_2)$, we can compute $LL(\gamma)$ in time $O\big(n^{1/3}\log^3 n \cdot(|LL(\gamma_1)|+|LL(\gamma_2)|)\big)$.
\end{lemma}

\begin{proof}
    We will sweep a ray through $C_B$, and maintain the invariant that the representation of $LL(\gamma)$ has been computed in the swept part of $C_B$. The events correspond to the points that appear in $LL(\gamma_1)\cup LL(\gamma_2)$.  We initialize the event queue by adding all points in $LL(\gamma_1)\cup LL(\gamma_2)$ in clockwise order; this takes $O(|LL(\gamma_1)|+|LL(\gamma_2)|)$ time.

    We sweep $C_B$ clockwise.. We keep track of the current curves intersected by the sweeping ray $\alpha_p$ corresponding to some triplet $(p,a_p,b_p)$ from $LL(\gamma_1)$ and some curve $\alpha_q$ corresponding to some triplet $(q,a_q,b_q)$ from $LL(\gamma_2)$. Let $C$ denote the section of $C_B$ between the event being processed $a'$ and the next event $a''$. We invoke the algorithm of~\Cref{lm:cross-point} to get a description of $\bd(F_p\cap F_q)\cap C_B$, which we restrict to $C$ and we append to $LL(\gamma)$ the corresponding entries. Note that we add two entries to $LL(\gamma)$ only if the procedure returns an overlap point $a$ that falls into $C$. The running time is dominated by the $O(|LL(\gamma_1)|+|LL(\gamma_2)|)$ invocations of~\Cref{lm:cross-point}, each of which takes $O(n^{1/3}\log^3 n)$ time.
\end{proof}

We are now ready to conclude the proof of~\Cref{lm:check-cells}.

\begin{proof}[Wrap-up of the proof of~\Cref{lm:check-cells}]
    We restrict $P$ to the cells $A$ and $B$, and restrict the set of disks that intersect the cells $A$ and $B$ in $O(n)$ time. Next, we compute the implicit representations of $F_p$ for each $p\in P\cap A$ as described above.
    
    Let $T$ be the minimum depth complete rooted binary tree with at least $|P\cap A|$ leaves. Assign the points of $P\cap A$ to distinct leaves, and for a node $v$ let $S_v\subset P\cap A$ denote the set of points assigned to the descendants of~$v$. We compute the curves $\gamma_{S_v}:=\bd\big( \bigcap_{p\in S_v} F_p \big)\cap C_B$ in a bottom-up order on $T$ using \Cref{lem:raysweep} each time. At the root we obtain an implicit representation of the curve $\gamma^*=\bd\big( \bigcap_{p\in P\cap A} F_p \big)\cap C_B$. Note that at each level of $T$ each flower $F_p$ contributes to exactly one set $S_v$, so the processing of each level takes at most $O(n\cdot n^{1/3}\log^3 n)$ time. Since $T$ has depth $O(\log n)$, the procedure takes $O(n^{4/3}\log^4 n)$ time.

    It remains to check if $P\cap B\subset \bigcap_{p\in P\cap A} F_p$. First, we copy the implicit representation of $\gamma^*$ into an array. Then, for each $q\in P\cap B$ observe that $q\in \bigcap_{p\in P\cap A} F_p$ if and only if the intersection of the ray $\overrightarrow{oq}$ with $\gamma^*$ occurs after $q$ along  $\overrightarrow{oq}$. To check this, we first use binary search on the implicit representation of $\gamma^*$ to find the triple $(p,a,b)$ such that $\overrightarrow{oq}$ is between $\overrightarrow{oa}$ and $\overrightarrow{ob}$ in the clockwise order. We use \Cref{lm:ray-shooting} to find the intersection point in $O(n^{1/3}\log^2 n)$ time. Thus, checking containment for each $q\in P\cap B$ takes $O(n^{4/3}\log^2 n)$ time. The total running time for the cell pair $A,B$ is therefore $O(n^{4/3}\log^4 n)$, as required.
\end{proof}

\bibliographystyle{alphaurl}
\bibliography{refs}

\newpage
\appendix
\section{Diameter Algorithm for Line Segments}
\label{ap:diameter-algorithms}
For fast computation of diameter in geometric graphs, Duraj, Lech, and Pot\c{e}pa~\cite{duraj2023better} first showed an algorithm for \Diameter-$\Delta$ in geometric intersection graphs that runs in $\OO(\Delta n^{2-1/d})$ time if the VC-dimension of $r$-neighborhood balls is bounded by $d$ and a certain geometric data structure exists. 
The data structure requirement was relaxed by \cite{ChanCGKLZ25} to be for \emph{rainbow colored intersection searching} (RIS), where every vertex in the geometric intersection graph has a color and the goal is to design a data structure to report whether the neighborhood of a vertex contains all the colors. (The number of colors can be arbitrary, in the extreme, there may be $\Omega(n)$ colors.) 
Their framework showed the following:

\begin{theorem}[Truly Subquadratic Diameter Framework~\cite{ChanCGKLZ25}]
\label{thm:diameter}
Let $G=(V,E)$ be a geometric intersection graph of a set of $n$ geometric objects. Suppose that these properties hold:
\begin{enumerate}
    \item[(a)] \ul{Decremental BFS Data Structure}: There exists a data structure with $\OO(n)$ preprocessing time that supports computing a neighborhood ball $N^{r}[v]$ in $\OO(|N^{r}[v]|)$ time and supports deletion of the vertices $N^{r}[v]$ from the graph in $\OO(|N^{r}[v]|)$ time.
    \item[(b)] \ul{Efficient RIS Data Structure}: Suppose every vertex of $V$ is associated with a color, there exists a data structure with $\OO(n)$ preprocessing time that can decide if a query vertex $v\in V$ has all the colors in its neighborhood $N[v]$ in $\OO(1)$ time.
    \item[(c)] \ul{Bounded distance VC-dimension}:
    The distance VC-dimension of $G$ is at most $d$.
\end{enumerate}
Then there is an algorithm for \Diameter-$\Delta$ in $\OO(\Delta\cdot n^{2-1/(2d)})$ time and an algorithm for computing the exact diameter in $\OO(n^{2-1/(4d)})$ time.
\end{theorem}

For general string graphs, it is likely impossible to obtain the data structure guarantees of (a) and (b). Indeed, it is believed that even for line segments, (a) and (b) require $\Omega(n^{4/3})$ total time to answer $O(n)$ queries via reductions from Hopcroft's problem~\cite{Erickson1996}. On the other hand, for $h$-slope line segments where $h$ is a constant, it is possible to implement (a) using standard orthogonal range searching techniques, and (b) by constructing a RIS data structure with respect to exactly one slope at a time via orthogonal range searching with standard techniques for handling colors~\cite{GuptaSURVEY}. 

\subsection{Diameter Algorithm for Non-degenerate Axis-aligned Line Segments}
\label{ssec:diameter-NDAALS}

For completeness, we briefly describe the strategy for an efficient RIS data structure for non-degenerate horizontal 
line segments. We emphasize that this reduction to orthogonal range searching is standard~\cite{GuptaSURVEY}.
Consider an input set of $n$ horizontal line segments $H$, each with a color. 
For every color $c$, take the vertical decomposition of the color-$c$ line segments $H_c\subseteq H$, and then consider each of the $O(|H_c|)$ rectangular regions of the vertical decomposition that is bounded from the positive $y$ direction by $h_y$. Map each such region to the 3D box with this rectangle at its base and extending infinitely in the positive $z$ direction starting at $z=h_y$.
Now, given a query vertical line segment from $(q_x, q_y)$ to $(q_x, q_y')$ with $y\le y'$, map this line segment to the 3D point $q = (q_x, q_y, q_y')$.
If the line segment intersects a line segment of $H_c$, then $q$ is contained in exactly one of the boxes created earlier from the vertical decomposition of $H_c$.
It suffices to count the boxes enclosing the query point $q$, which can be done in $\OO(n)$ space and $\OO(1)$ query time via orthogonal range searching~\cite{AgarwalE99,CGAA}.

Combining this with orthogonal range searching for decremental BFS and the VC-dimension bound of \Cref{thm:string-bipartite}, we obtain the following result using \Cref{thm:diameter}.
\GenPosLineSeg*

\subsection{Diameter Algorithm for $h$-Slope Line Segments}
\label{ap:h-slope-line-seg}

Let $G=(V,E)$ be the intersection graph of $h$-slope line segments with diameter $\Delta=O(1)$ and $h=O(1)$. For any segment $v\in V$, let $\chi(v)$ denote the slope of $v$. We will also call $\chi(v)$ the color of $v$, and consider color sequences $S$ consisting of length at most $\Delta+1$ with the colors in $[h]$. We will use array notation $S[i]$ to denote the color at the $i$th index of $S$, and $S[a:b]$ to denote the subarray of $S$ from the $a$th position to the $b$th position, and for two sequences $S_1$ and $S_2$, we denote the concatenation of the sequences by $S_1\circ S_2$.
Recall that we use the following notation to define balls of $S$ centered at $v$.
\[
B_S(v) = \{u\in V: \text{ $\exists S'\sqsubseteq S$ and an $S'$-rainbow path $P(v\rightarrow u)$}\}
\]

\paragraph{Computing a good ordering.}
We show how to compute a compact representation of $B_S(v)$ for all sequences $S$ and all vertices $v$. 
To be precise, let $\lambda$ be an ordering of $V$, and the \EMPH{$\lambda$-interval representation $\Rep_\lambda(B_S(v))$} as a collection of maximal contiguous subsequences of $\lambda$ called \EMPH{intervals}, whose union is $B_S(v)$.
By \cite[Lemma 2.10]{ChanCGKLZ25}, there exists an ordering $\lambda$ (called a stabbing path) such that 
\[\sum_{v\in V}\sum_{S\in\cS}|\Rep_\lambda(B_S(v))| = \OO(|V||\cS|\cdot n^{1-1/d}) = O(n^{2-1/d}),\]
where $d$ is the dual VC-dimension of $(V, \{B_S(v)\}_{v\in V, S\in \cS})$. Note that $|\cS| = O(h^{\Delta+1}) = O(1)$. Furthermore, this can be computed in $\OO(n^{1+1/d})$ time. Note that by \Cref{lm:VCbound-type} and \Cref{thm:segments}, we have that $d \le 4|\cS| = O(h^{\Delta+1})$.

\paragraph{Ball growing for sequences.}
Fix a color sequence $S\subseteq\cS$ where $\cS$ consists of all sequences of $h$ colors with length at most $\Delta+1$.
Let $v$ be an arbitrary vertex. Without loss of generality, we may assume that $\chi(v) = S[1]$, or else we could delete the prefix of $S$ that does not contain $\chi(v)$.
The following identity follows from the definition of $B_S(v)$; we will call this the \EMPH{ball-growing equation}.
\begin{equation}
\label{eq:ball-growing}
\tag{ball-growing equation}
B_{S}(v) = B_{S[1]\circ S[3:s]}(v) \cup \bigcup_{\substack{w\in N[v]\\ \chi(w) = S[2]}} B_{S[2:s]}(w)
\end{equation}
The interpretation of this equation is that vertices reachable by $S$ either is via a rainbow path with a subsequence of $S[1] \circ S[3:s]$, or consists of $v$ followed by a neighbor $u\in N[v]$ with $\chi(u) = S[2]$ as the start of a rainbow path with a subsequence $S[2:s]$.

We remark that $B_{S[1]\circ S[3:s]}(v)$ and each individual $B_{S[2:s]}(u)$ have compact interval representations under the ordering $\lambda$ (as they are balls of the form $B_{S'}(u)$ for some sequence $S'$), but the (partial) union in the second term of the \ref{eq:ball-growing} does not.
Critically the union is only over a single color of the neighbors of $u$ that we consider.
For this reason, we have to implement the ball growing procedure by constructing a data structure to the following problem that is a generalization of \EMPH{Interval Searching} in \cite[Problem 2.11]{ChanCGKLZ25}. 
\begin{problem}[Extended Interval Searching]\label{prob:DS0} Let $\mathcal{O}_{IS}$ be a given set of objects, where each object $o\in \mathcal{O}_{IS}$ is associated with a set of intervals (of integer points) of $[1:n]$, denoted by~\EMPH{$\mathcal{I}_o$}.  
Design a data structure that answers the following query: 
\begin{itemize}
    \item \textsc{IntervalSearch}$(q)$:  Given an object $q$ (not necessarily in $\mathcal{O}_{IS}$) associated with a set of input intervals of $[1:n]$, denoted by \EMPH{$\mathcal{I}_{in}(q)$}, 
    return (the interval representation of) all the integer points in $[1:n]$ associated with objects in $\mathcal{O}_{IS}$ that intersect\footnote{Here we mean the objects intersect, not their associated intervals.} $q$, unioned with $\mathcal{I}_{in}(q)$.
\end{itemize}
\end{problem}
We defer the proof of the following lemma till \Cref{ss:IC-reductions}.
\begin{lemma}\label{lm:EIS}
For any parameter $b\ge 1$, and any fixed sequence $S$, there is a data structure for Extended Interval Searching for $h$-slope line segments where the input is (the $\lambda$-representation of) $\{B_S(v)\}_{v\in V: \chi(w)=S[1]}$ that runs in $\OO(b\cdot N_{ES} + L_{ES}/b)$ where $N_{ES}$ is the number of input and query intervals and $L_{ES}$ is the total length of the input and query intervals.

\end{lemma}
Using \Cref{lm:EIS} we can prove the following lemma.
\begin{lemma}\label{lm:ball_growing_segments}
There exists an algorithm for computing $\Rep_\lambda(B_S(v))$ for all $v\in V$ and all $S\in \cS$ that runs in time $\OO(n^{2-1/(2d)})$, where $\sum_{S\in S, v\in V}|\Rep_\lambda(B_S(v))| = \OO(n^{2-1/d})$.
\end{lemma}

\begin{proof}
We consider sequences in $\cS$ in order of length.
For every sequence $S$, we will compute $B_S(v)$ for all vertices $v\in V$ where $\chi(v) = S[1]$ via dynamic programming (recall that this computation is redundant for all other $S$).

Fix a sequence $S$ of length $s$ that we wish to compute for. 
We will compute $B_{S}(v)$ for all $v$ by applying the ball-growing equation. 
Note that $S[1]\circ S[3:s]$ is a shorter subsequence so we have already computed an interval representation of $B_{S[1]\circ S[3:s]}(v)$. 
Furthermore, by applying \Cref{lm:EIS} we can build a data structure for extended interval searching for the set $\Set{\big. B_{S[2:s]}(w) : \text{for all $w\in V$ with $\chi(w) = S[2]$}}$, and $\cI_{in}(q) = \Rep(B_{S[1]\circ S[3:s]}(v))$ for all $v$. 
The total runtime is $\OO(b\cdot \sum_{S\in \cS}\sum_{v\in V}|\Rep_\lambda(B_S(v))|+n^2/b) = O(n^{2-1/(2d)})$, by setting $b = O(n^{1/(2d)})$.
\end{proof}

With these balls, it is straightforward to compute $N^{\Delta}[v]$ for all $v\in V$ by directly taking the union of $B_S(v)$ over all subsequences $S$ of length at most $\Delta+1$ in total time $\sum_{S\in S, v\in V}|\Rep_\lambda(B_S(v))| = \OO(n^{2-1/d})$, thereby proving \Cref{thm:h-slope-line-seg} with total running time $n^{2-1/O(h^{\Delta+1})}$ as $d= O(h^{\Delta+1})$.

\subsection{Interval Cover Reductions}
\label{ss:IC-reductions}

To prove \Cref{lm:EIS}, we follow the chain of reductions between data structure problems in \cite{ChanCGKLZ25}, beginning with the following problem.

\begin{problem}
[Interval Cover Problem {\cite[Problem 1.3]{ChanCGKLZ25}}] 
\label{def:DS-1} 
Given a set of $N$ objects $\mathcal{O}$ and each object $o\in \mathcal{O}$ is  associated with an interval $I_o\subseteq [1:n]$. Design a data structure to answer the following query:
\begin{itemize}
    \item \textsc{Covers?}$(q,I)$:  Given a query object $q$ and a query interval  $I\subseteq [1:n]$, decide whether the union of intervals associated with the objects intersecting $q$ in  $\mathcal{O}$ covers the whole $I$.  
\end{itemize}
\end{problem}

We give a modified version of \cite[Lemma 2.12]{ChanCGKLZ25}, with a modified proof.

\begin{lemma}
\label{lm:DS1-to-DS0}
If we can construct a data structure $\mathcal{D}_{IC}$ for solving Problem~\ref{def:DS-1} with preprocessing time $P(N)$ and query time $Q(N)$, then we can construct a data structure $\mathcal{D}_{IS}$  for solving Problem~\ref{prob:DS0} with
preprocessing time $\OO(P(\tilde{N}_{IS}))$ and query time  $\OO( (|\mathcal{I}_{in}(q)| + |\mathcal{I}_{out}(q)|) \cdot Q(\tilde{N}_{IS}))$  where $\EMPH{$\tilde{N}_{IS}$} \coloneqq \sum_{o\in \mathcal{O}_{IS}}|\mathcal{I}_o|$ is the total number of input intervals and $\mathcal{I}_{out}(q)$ is the set of output intervals from the extended interval search query of $q$ to $\mathcal{D}_{IS}$. 
\end{lemma}
\begin{proof}
We take an instance of Problem~\ref{prob:DS0}. For each object $s\in \mathcal{O}_{IS}$ we duplicate it to $k$ copies if $s$ is associated with $k$ intervals. Each copy is now associated with a single interval of $\mathcal{I}_o$. This creates a total of $\tilde{N}_{IS}=\sum_{o\in \mathcal{O}_{IS}}|\mathcal{I}_o|$ objects.
Now we build a data structure $\mathcal{D}_{IC}$ to solve Problem~\ref{def:DS-1} on this set of duplicated objects with preprocessing time $\OO(P(\tilde{N}_{IS}))$. 
For each query \textsc{IntervalSearch}$(q)$, we recursively issue queries to $\mathcal{D}_{IC}$. Specifically, we start with $I \gets [1 \,..\, n]$. 
If every subinterval of $I\setminus \bigcup_{J\in\mathcal{I}_{in}(q)}J$ is completely covered by the union of the intervals associated with objects in $\mathcal{O}_{IS}$ that intersect $q$ (which is checked by a query to $\mathcal{D}_{IC}$ with $q$ and each such subinterval), we output $I$ and we are done. Otherwise, if $I$ is completely avoided, we output $\varnothing$ and we are also done. 
For the remaoning case, we will recurse:  We divide $I$ into two intervals of equal length, $I_1$ and $I_2$, and issue queries \textsc{Covers?}$(q, I_1)$ and \textsc{Covers?}$(q, I_2)$ with $\mathcal{D}_{IC}$. 
In the end, we will output the union of all the intervals that are fully covered by the intervals associated with objects in $\mathcal{O}_{IS}$ that intersect $q$.

The running time for a query $q$ is dependent on the total number of queries issued to $\mathcal{D}_{IC}$ recursively. Notice that all query intervals are dyadic intervals.
In addition, recursion stops when an interval $I$ is completely covered by the union of intervals $\mathcal{S}(q)$ or completely avoided. Thus only the dyadic intervals whose parent partially overlaps with a query output interval will ever trigger a query. The total number of such intervals is in the order of $O(|\mathcal{I}_{out}(q)|\cdot \log n)$. 
The number of subintervals that we query can be charged to endpoints in $\mathcal{I}_{in}(q)$ per level, and so is $O((|\mathcal{I}_{in}(q)|+|\mathcal{I}_{out}(q)|)\cdot \log n)$ in total.
Recall that each query to $\mathcal{D}_{IC}$ takes time  $Q(\tilde{N}_{IS})$. Summing up everything, we have the claim in the Lemma. 
\end{proof}

Now we can prove \Cref{lm:EIS}.

\begin{proof}[Proof of \Cref{lm:EIS}]
By \cite[Lemma 2.14]{ChanCGKLZ25},
to solve Problem~\ref{prob:DS0} in the time desired for \Cref{lm:DS1-to-DS0},
it suffices to show we can construct an RIS data structure for objects corresponding to (possibly multiple copies of) the vertices $\{w\in V: \chi(w)=S[1]\}$ 
with construction time near-linear in the number of objects and $\OO(1)$ query time. 
We remark that all these vertices correspond to line segments with the same slope $S[1]$.

If we fix the slope of the query line segment, an efficient RIS data structure for these line segments with the same slope can be constructed analogous to the non-degenerate orthogonal line segment RIS data structure described in \Cref{sec:segs} (if the input and query intervals have the same slope, the problem instead reduce to a 1D colored range counting problem whose solution is standard, see \cite{GuptaSURVEY}). 
This data structure can be constructed in near-linear time in the number of segments and has $\OO(1)$ query time as desired.
\end{proof}

\end{document}